\documentclass[a4paper]{article}

\usepackage[top=1.5in, bottom=1.5in, left=1.25in, right=1.25in]{geometry} 

\usepackage{comment}

\usepackage{amssymb}
\usepackage{graphicx}
\usepackage{latexsym}
\usepackage{amssymb}
\usepackage{amsmath}
\usepackage{amsfonts}
\usepackage{amsthm}
\usepackage{cite}
\usepackage{algorithmic}
\usepackage{algorithm}
\usepackage{textcomp}
\usepackage{color}
\usepackage{mathtools}
\usepackage{stmaryrd}
\usepackage{tikz}
\usetikzlibrary{calc}
\usepackage[normalem]{ulem}
\usepackage{url}
\usepackage{array}
\usepackage{hyperref}
\usepackage{authblk}

\DeclareMathOperator{\oneD}{\texttt{1D}}
\DeclareMathOperator{\undef}{\texttt{undef}}

\newcommand{\transpose}[1]{{#1}^{\texttt{T}}}

\newcommand{\ta}{\ensuremath{\mathtt{a}}}
\newcommand{\tb}{\ensuremath{\mathtt{b}}}
\newcommand{\tc}{\ensuremath{\mathtt{c}}}
\newcommand{\td}{\ensuremath{\mathtt{d}}}
\newcommand{\te}{\ensuremath{\mathtt{e}}}
\newcommand{\tf}{\ensuremath{\mathtt{f}}}
\newcommand{\tg}{\ensuremath{\mathtt{g}}}
\newcommand{\tth}{\ensuremath{\mathtt{h}}}

\definecolor{darkgreen}{rgb}{0.0, 0.7, 0.0}
\definecolor{darkorange}{rgb}{1.0, 0.3, 0.0}
\definecolor{darkyellow}{rgb}{0.9, 0.9, 0.0}
\definecolor{lightgrey}{rgb}{0.8, 0.8, 0.8}
\definecolor{MyColor}{rgb}{1, 0, 1}

\newcounter{row}
\newcounter{col}

\newtheorem{example}{Example}
\newtheorem{remark}{Remark}
\newtheorem{definition}{Definition}
\newtheorem{lemma}{Lemma}
\newtheorem{theorem}{Theorem}
\newtheorem{corollary}{Corollary}
\newtheorem{proposition}{Proposition}

\begin{document}

\title{Two-Dimensional Pattern Languages\thanks{This document is a full version (i.\,e., it contains all proofs) of the conference paper \cite{FerSchSub2013}.}}

\author[1]{Henning Fernau}
\author[1]{Markus~L.~Schmid}
\author[2]{K.~G.~Subramanian}

\affil[1]{Fachbereich 4 -- Abteilung Informatikwissenschaften, Universit\"at Trier, 54286 Trier, Germany, \texttt{\{fernau,mschmid\}@uni-trier.de}}
\affil[2]{School of Computer Sciences, Universiti Sains Malaysia, 11800 Penang, Malaysia, \texttt{kgsmani1948@yahoo.com}}

\maketitle

\begin{abstract}
We introduce several classes of array languages obtained by generalising Angluin's pattern languages to the two-dimensional case. These classes of two-dimensional pattern languages are compared with respect to their expressive power and their closure properties are investigated.
\end{abstract}

\section{Introduction}\label{sec:intro}

Several methods of generation of \emph{two-dimensional languages} (also called \emph{array languages} or \emph{picture languages}) 
have been proposed in the literature, extending the techniques and results of formal string language theory. 
A picture is considered as a rectangular array of terminal symbols in the two-dimensional plane. 
Models based on grammars or automata as well as those based on theoretical properties of the string languages are well-known and have been extensively investigated. We refer the interested readers to books and surveys like the ones by  Rosenfeld~\cite{ros:pic}, Wang~\cite{wan:arr}, Rosenfeld~and~Siromoney~\cite{ros:pic2}, Giammarresi~and~Restivo~\cite{gia:two}, or Morita~\cite{mor:two}.  
For example, regular \emph{string} languages (also known as recognizable string languages) can be characterized in terms of local languages and projections.
Based on a similar idea, 
the class REC of recognizable \emph{picture} languages (see Giammarresi~and~Restivo~\cite{gia:rec}) was proposed as a two dimensional counterpart of regular string languages. 
In this work, we attempt to generalise a class of string languages to the two-dimensional case, which also provides several desirable features and  has therefore attracted considerable interest over the last three decades
in the formal language theory community as well as in the learning theory community: Angluin's pattern languages (see \cite{ang:fin2}). \par
In this context, a \emph{pattern} is a string over an alphabet $\{x_1, x_2, x_3, \ldots\}$ of \emph{variables}, e.\,g., $\alpha := x_1\,x_1\,x_2\,x_2\,x_1$. For some finite alphabet $\Sigma$ of \emph{terminal symbols}, the \emph{pattern language} described by $\alpha$ (with respect to $\Sigma$) is the set of all words over $\Sigma$ that can be derived from $\alpha$ by uniformly substituting the variables in $\alpha$ by (non-empty) terminal words. For example, if $\Sigma := \{\ta, \tb, \tc\}$, then $u := \tb \tc \tb\tb \tc \tb\tc \tc \ta\tc \tc \ta\tb \tc \tb$ and $v := \ta  \tb\ta  \tb\ta  \tb  \ta\ta  \tb  \ta\ta  \tb$ are words of the pattern language given by $\alpha$, since replacing $x_1$ by $\tb  \tc  \tb$ and $x_2$ by $\tc  \tc  \ta$ turns $\alpha$ into $u$ and replacing $x_1$ by $\ta  \tb$ and $x_2$ by $\ta  \tb  \ta$ turns $\alpha$ into $v$. On the other hand, the word $\tc \ta \tb \ta \tc \ta \tb \ta \tb \ta \tb \ta$ is not a member of the pattern language of $\alpha$.\par
One of the most notable features of pattern languages is that they have natural and compact human readable descriptors (or generators), namely the patterns. In particular, this advantage becomes evident when patterns are compared to other language descriptors as, e.\,g., grammars or automata, which are usually quite involved even though the language they describe is rather simple. Nevertheless, patterns can compete with common automata models and grammars in terms of expressive power and their practical relevance is demonstrated by the widespread use of so-called extended regular expressions with backreferences, which implicitly use the concept of patterns and are capable of defining all pattern languages.\footnote{In fact, these extended regular expressions with backreferences are nowadays a standard element of most text editors and programming languages (cf. Friedl~\cite{fri:mas}). } \par
The main goal of this paper is to generalise the concept of patterns as language descriptors to the two-dimensional case, while preserving the desirable features of (one-dimensional) pattern languages, i.\,e., the simplicity and compactness of their descriptors. The work done so far on two-dimensional languages demonstrates that there are difficulties that seem to be symptomatic for the task of generalising a class of string languages to the two-dimensional case. Firstly, such a generalisation is usually accompanied with a substantial increase in complexity of the descriptors (e.\,g., when extending context-free or contextual grammars to the two-dimensional case (see Fernau~et~al.~\cite{FerFreHol98a}, Freund~et~al.~\cite{fre:con})) and, secondly, there are often many competing and seemingly different ways to generalise a specific class of string languages, which all can be considered natural (e.\,g., it is still on debate what the appropriate two-dimensional counterpart of the class of regular languages 
might be (see Giammarresi~et~al.~\cite{Giaetal96}, Matz~\cite{mat:rec})). Our two-dimensional patterns, to be introduced in this work, are as simple and compact as their one-dimensional counterparts. Although there are several different possibilities of how these two-dimensional patterns can describe two-dimensional languages, one of these sticks out as the intuitively most natural one. Hence, the model of Angluin's pattern languages seems to be comparatively two-dimensional friendly.\par
Besides the conceptional contribution of this paper, we present a comparison between the expressive power of different classes of two-dimensional pattern languages and an investigation of their closure properties. We conclude the paper by outlining further research questions and possible extensions to the model of two-dimensional pattern languages.

\section{Preliminaries}\label{sec:def}

In this section, we briefly recall the standard definitions and notations regarding one- and two-dimensional words and languages.\par
Let $\mathbb{N} := \{1, 2, 3, \ldots \}$ and let $\mathbb{N}_{0} := \mathbb{N} \cup \{0\}$. For a finite alphabet $\Sigma$, a \emph{string} or \emph{word} (\emph{over $\Sigma$}) is a finite sequence of symbols from $\Sigma$, and $\varepsilon$ stands for the \emph{empty string}. The notation $\Sigma^+$ denotes the set of all nonempty strings over $\Sigma$, and $\Sigma^*:=\Sigma^+ \cup \{ \varepsilon \}$. For the \emph{concatenation} of two strings $w_1, w_2$ we write $w_1 \cdot w_2$ or simply $w_1 w_2$. We say that a string $v \in \Sigma^*$ is a \emph{factor} of a string $w \in \Sigma^*$ if there are $u_1, u_2 \in \Sigma^*$ such that $w = u_1 \cdot v \cdot u_2$. If $u_1$ or $u_2$ is the empty string, then $v$ is a \emph{prefix} (or a \emph{suffix}, respectively) of $w$. The notation $|w|$ stands for the length of a string $w$.\par
A \emph{two-dimensional word} (or \emph{array}) \emph{over $\Sigma$} is a tuple 
\begin{equation*}
W := ((a_{1,1}, a_{1,2}, \ldots, a_{1,n}), (a_{2,1}, a_{2,2}, \ldots, a_{2,n}), \ldots, (a_{m,1}, a_{m,2}, \ldots, a_{m,n}))\,, 
\end{equation*}
where $m, n \in \mathbb{N}$ and, for every $i$, $1 \leq i \leq m$, and $j$, $1 \leq j \leq n$, $a_{i, j} \in \Sigma$. We define the \emph{number of columns} (or \emph{width}) and \emph{number of rows} (or \emph{height}) of $W$ by $|W|_c := n$ and $|W|_r := m$, respectively. The \emph{empty array} is denoted by $\lambda$, i.\,e., $|\lambda|_c = |\lambda|_r = 0$. For the sake of convenience, we also denote $W$ by $[a_{i, j}]_{m, n}$ or by a matrix of one of the following forms:
\begin{equation*}
\begin{smallmatrix}
  a_{1,1} & a_{1,2} & \ldots & a_{1,n} \\
  a_{2,1} & a_{2,2} & \ldots & a_{2,n} \\
  \vdots & \vdots & \ddots & \vdots \\
  a_{m,1} & a_{m,2} & \ldots & a_{m,n} 
\end{smallmatrix}
,\:\:\:\:\:\:
\begin{bsmallmatrix}
  a_{1,1} & a_{1,2} & \ldots & a_{1,n} \\
  a_{2,1} & a_{2,2} & \ldots & a_{2,n} \\
  \vdots & \vdots & \ddots & \vdots \\
  a_{m,1} & a_{m,2} & \ldots & a_{m,n} 
\end{bsmallmatrix}\,.
\end{equation*}
If we want to refer to the $j^{\text{th}}$ symbol in row $i$ of the array $W$, then we use $W[i, j] = a_{i, j}$. By $\Sigma^{++}$, we denote the set of all nonempty arrays over $\Sigma$, and $\Sigma^{**}:=\Sigma^{++} \cup \{ \lambda \}$. Every subset $L \subseteq \Sigma^{**}$ is an \emph{array language}. \par
Let $W := [a_{i, j}]_{m, n}$ and $W' := [a'_{i, j}]_{m', n'}$ be two non-empty arrays over $\Sigma$. 
The \emph{column concatenation} of $W$ and $W'$, denoted by $W \varobar W'$, is undefined if $m \neq m'$ and is the array
\begin{equation*}
\begin{smallmatrix}
  a_{1,1} & a_{1,2} & \ldots & a_{1,n} & b_{1,1} & b_{1,2} & \ldots & b_{1,n'} \\
  a_{2,1} & a_{2,2} & \ldots & a_{2,n} & b_{2,1} & b_{2,2} & \ldots & b_{2,n'}  \\
  \vdots & \vdots & \ddots & \vdots & \vdots & \vdots & \ddots & \vdots \\
  a_{m,1} & a_{m,2} & \ldots & a_{m,n} &  b_{m',1} & b_{m',2} & \ldots & b_{m',n'} 
\end{smallmatrix}
\end{equation*}
otherwise. The \emph{row concatenation} of $W$ and $W'$, denoted by $W \varominus W'$, is undefined if $n \neq n'$ and is the array
\begin{equation*}
\begin{smallmatrix}
  a_{1,1} & a_{1,2} & \ldots & a_{1,n} \\
  a_{2,1} & a_{2,2} & \ldots & a_{2,n} \\
  \vdots & \vdots & \ddots & \vdots \\
  a_{m,1} & a_{m,2} & \ldots & a_{m,n} \\
  b_{1,1} & b_{1,2} & \ldots & b_{1,n'} \\
  b_{2,1} & b_{2,2} & \ldots & b_{2,n'} \\
  \vdots & \vdots & \ddots & \vdots \\
  b_{m',1} & b_{m',2} & \ldots & b_{m',n'} 
\end{smallmatrix}
\end{equation*}
otherwise. Intuitively speaking, the vertical line and the horizontal line in the symbols $\varobar$ and $\varominus$, respectively, indicate the edge where the arrays are concatenated. In order to denote that, e.\,g., $U \varominus V$ is undefined, we also write $U \varominus V = \undef$. 

Furthermore, for every array $U$, $U \varobar \lambda = \lambda \varobar U = U \varominus \lambda = \lambda \varominus U = U$ and 
$U \varobar \undef = \undef \varobar U = U \varominus \undef = \undef \varominus U = \undef$. Algebraically speaking, if $\Sigma^{**}_{\undef}:=\Sigma^{**}\cup\{\undef\}$, then
$(\Sigma^{**}_{\undef},\varobar,\lambda)$ and $(\Sigma^{**}_{\undef},\varominus,\lambda)$ both form monoids with $\undef$ as an absorbing element.

\begin{example}\label{concatExample1}
Let 
\begin{equation*}
W_1 := 
\begin{bsmallmatrix}
  \ta & \tb & \ta \\
  \tb & \tc & \ta \\
  \ta & \tb & \tb \\
 \end{bsmallmatrix}
, W_2 := 
\begin{bsmallmatrix}
  \tb & \tc \\
  \tb & \ta \\
  \tc & \ta
 \end{bsmallmatrix}
, W_3 := 
\begin{bsmallmatrix}
  \ta & \tb & \tc \\
  \tc & \tb & \tb
 \end{bsmallmatrix}
\text{ and } W_4 := 
\begin{bsmallmatrix}
  \ta & \ta \\
  \ta & \tb
 \end{bsmallmatrix}\,.
\end{equation*}
Then $W_1 \varominus W_2$ and $W_1 \varobar W_3 = \undef$, but 
\begin{equation*}
W_1 \varobar W_2 = 
\begin{bsmallmatrix}
\ta & \tb & \ta & \tb & \tc \\
  \tb & \tc & \ta &  \tb & \ta \\
  \ta & \tb & \tb &  \tc & \ta
 \end{bsmallmatrix}
\text{ and } W_1 \varominus W_3 =
\begin{bsmallmatrix}
  \ta & \tb & \ta \\
  \tb & \tc & \ta \\
  \ta & \tb & \tb \\
  \ta & \tb & \tc \\
  \tc & \tb & \tb 
\end{bsmallmatrix}\,.
\end{equation*}
\end{example}

\begin{example}
Let $W_1, W_2, W_3$ and $W_4$ be defined as in Example~\ref{concatExample1}.
\begin{align*}
&(W_1 \varobar W_2) \varominus (W_3 \varobar W_4) = (W_1 \varominus W_3) \varobar (W_2 \varominus W_4) =
\begin{bsmallmatrix}
  \ta & \tb & \ta & \tb & \tc \\
  \tb & \tc & \ta & \tb & \ta \\
  \ta & \tb & \tb & \tc & \ta \\
  \ta & \tb & \tc & \ta & \ta \\
  \tc & \tb & \tb & \ta & \tb
 \end{bsmallmatrix}\,,\\
&(W_1 \varobar W_2) \varominus (W_4 \varobar W_3) = 
\begin{bsmallmatrix}
  \ta & \tb & \ta & \tb & \tc \\
  \tb & \tc & \ta & \tb & \ta \\
  \ta & \tb & \tb & \tc & \ta \\
  \ta & \ta & \ta & \tb & \tc \\
  \ta & \tb & \tc & \tb & \tb
 \end{bsmallmatrix}\,,\:\:
(W_1 \varominus W_4) \varobar (W_2 \varominus W_3) = \undef\,.
\end{align*}
\end{example}

Next, we define some operations for array languages. The row and column concatenation for array languages $L_1$ and $L_2$ is defined by $L_1 \varominus L_2 := \{U \varominus V \mid U \in L_1, V \in L_2, U \varominus V \neq \undef\}$ and $L_1 \varobar L_2 := \{U \varobar V \mid U \in L_1, V \in L_2, U \varobar V \neq \undef\}$, respectively. For an array language $L$ and $k \in \mathbb{N}$, $L^{\varominus k}$ denotes the $k$-fold row concatenation of $L$, i.\,e., $L^{\varominus k} := L_1 \varominus L_2 \varominus \ldots \varominus L_k$, $L_i = L$, $1 \leq i \leq k$. The $k$-fold column concatenation, denoted by $L^{\varobar k}$, is defined analogously. The \emph{row} and \emph{column concatenation closure} of an array language $L$ is defined by $L^{\varominus *} := \bigcup_{k = 1}^{\infty} L^{\varominus k}$ and $L^{\varobar *} := \bigcup_{k = 1}^{\infty} L^{\varobar k}$, respectively. Obviously, the row and column concatenation closure of an array language correspond to the Kleene closure of a string language. \par
Now, we turn our attention to some geometric operations for arrays. The \emph{transposition} of an array $U$, denoted by $\transpose{U}$, is obtained by reflecting $U$ along the main diagonal. 
The \emph{$\varominus$-reflection} and \emph{$\varobar$-reflection} of $U$, denoted by $U^{\varominus\texttt{R}}$ and $U^{\varobar\texttt{R}}$, respectively, 
are obtained by reflecting $U$ along the horizontal and vertical axis, respectively. 
The \emph{right turn} and \emph{left turn} of $U$, denoted by $U^{\curvearrowright}$ and $U^{\curvearrowleft}$, respectively, is obtained by turning $U$ through $90$ degrees to the right and to the left, respectively. For example, if $U := 
\begin{bsmallmatrix}
  \ta & \tb & \tc & \td \\
  \te & \tf & \tg & \tth 
\end{bsmallmatrix}$, 
then 
\begin{align*}
\transpose{U} =
\begin{bsmallmatrix}
  \ta & \te \\
  \tb & \tf \\
  \tc & \tg \\
  \td & \tth \\
\end{bsmallmatrix}\,,
U^{\varominus\texttt{R}} =
\begin{bsmallmatrix}
  \te & \tf & \tg & \tth \\
  \ta & \tb & \tc & \td \\
\end{bsmallmatrix}\,,
U^{\varobar\texttt{R}} =
\begin{bsmallmatrix}
  \td & \tc & \tb & \ta \\
  \tth & \tg & \tf & \te \\
\end{bsmallmatrix}\,,
U^{\curvearrowright} =
\begin{bsmallmatrix}
  \te & \ta \\
  \tf & \tb \\
  \tg & \tc \\
  \tth & \td \\
\end{bsmallmatrix}\,,
U^{\curvearrowleft} =
\begin{bsmallmatrix}
  \td & \tth \\
  \tc & \tg \\
  \tb & \tf \\
  \ta & \te \\
\end{bsmallmatrix}\,,
(U^{\curvearrowright})^{\curvearrowright} =
\begin{bsmallmatrix}
  \tth & \tg & \tf & \te \\
  \td & \tc & \tb & \ta 
\end{bsmallmatrix}\,.
\end{align*}
We address left and right turn also as \emph{quarter-turns} below.
Moreover, the twofold right turn (displayed right-most in the example above) 
is also known as a \emph{half-turn}. 
\par 
A special operation considered in the context of arrays is the \emph{conjugation} of an array $U \in \Sigma^{**}$ with $|\Sigma| = 2$, denoted by $U^{\texttt{C}}$, which means that the two symbols of $\Sigma$ are exchanged in $U$, e.\,g., 
$\left(\begin{bsmallmatrix}
  \ta & \tb & \ta & \ta \\
  \tb & \ta & \ta & \tb \\
\end{bsmallmatrix}\right)^{\texttt{C}} = 
\begin{bsmallmatrix}
  \tb & \ta & \tb & \tb \\
  \ta & \tb & \tb & \ta \\
\end{bsmallmatrix}$. 
This can be also viewed as a quite restricted form of a two-dimensional morphism defined in the next section.\par 
All these operations for arrays are extended to array languages in the obvious way. \par
Next, we briefly summarise the concept of (one-dimensional) pattern languages as introduced in \cite{ang:fin2} by Angluin. Technically, the version of pattern languages used here are called \emph{nonerasing terminal-free} pattern languages (for an overview of different versions of one-dimensional pattern languages, the reader is referred to \cite{mat:pat} by Mateescu and Salomaa).\par
A (one-dimensional) \emph{pattern} is a string over an alphabet $X := \{x_1, x_2, x_3, \ldots\}$ of \emph{variables}, e.\,g., $\alpha := x_1\,x_1\,x_2\,x_2\,x_1$. In Section~\ref{sec:intro}, we have seen an intuitive definition of the language described by a pattern $\alpha$. This intuition can be formalised in an elegant way by using the concept of (\emph{word}) \emph{morphisms}, i.\,e., mappings $h : \Sigma_1^+ \rightarrow \Sigma_2^+$, which satisfy $h(u \, v) = h(u)\,h(v)$, for all $u, v \in \Sigma_1^+$. In this regard, for some finite alphabet $\Sigma$, the (one-dimensional) \emph{pattern language} of $\alpha$ (with respect to $\Sigma$) is the set $L^{\oneD}_{\Sigma}(\alpha) := \{h(\alpha) \mid h : X^+ \rightarrow \Sigma^+ \text{ is a morphism}\}$. An alternative, yet equivalent, way to define pattern languages is by means of factorisations. To this end, let $\alpha := y_1\,y_2 \ldots y_n$, $y_i \in X$, $1 \leq i \leq n$. Then $L^{\oneD}_{\Sigma}(\alpha)$ is the set of all words $w \in \Sigma^+$ that have a \emph{characteristic} factorisation for $\alpha$, i.\,e., a factorisation $w = u_1\,u_2 \cdots u_n$, such that, for every $i$, $1 \leq i \leq j \leq n$, $y_i = y_j$ implies $u_i = u_j$. It can be easily seen, that these two definitions are equivalent. However, for the two-dimensional case, we shall see that a generalisation of these two approaches will lead to different versions of two-dimensional pattern languages. The class of all one-dimensional pattern languages over the alphabet $\Sigma$ is denoted by $\mathcal{L}^{\oneD}_{\Sigma}$. 
We recall the example pattern $\alpha = x_1\,x_1\,x_2\,x_2\,x_1$ and the words $u := \tb \tc \tb\tb \tc \tb\tc \tc \ta\tc \tc \ta\tb \tc \tb$ and $v := \ta  \tb\ta  \tb\ta  \tb  \ta\ta  \tb  \ta\ta  \tb$ of Section~\ref{sec:intro} Since $h(\alpha) = u$ and $g(\alpha) = v$, where $h$ and $g$ are the morphisms induced by $h(x_1) := \tb  \tc  \tb$, $h(x_2) :=  \tc  \tc  \ta$ and $g(x_1) := \ta  \tb$, $g(x_2) := \ta  \tb  \ta$, we can conclude that $u, v \in L^{\oneD}_{\Sigma}(\alpha)$, where $\Sigma := \{\ta, \tb, \tc\}$.

\section{Two-Dimensional Pattern Languages}\label{sec:2DPatternLanguages}

As already mentioned, this work deals with the task of generalising pattern languages from the one-dimensional to the two-dimensional case. In order to motivate our approach to solve this task, we first spent some effort on illustrating the general difficulties and obstacles that arise.\par
Abstractly speaking, a pattern language for a given pattern $\alpha$ is the collection of all elements that satisfy $\alpha$. Thus, a sound definition of how elements satisfy patterns directly entails a sound definition of a class of pattern languages. In the one-dimensional case, the situation that a word satisfies a pattern is intuitively clear and it can be defined in several equivalent ways, i.\,e., a word $w$ satisfies the pattern $\alpha$ if and only if 
\begin{itemize}
\item $w$ can be derived from $\alpha$ by uniformly substituting the variables in $\alpha$,
\item $w$ is a morphic image of $\alpha$,
\item $w$ has a characteristic factorisation for $\alpha$.
\end{itemize}
We shall now demonstrate that for a two-dimensional pattern, i.\,e., a two-dimensional word over the set of variables $X$, e.\,g., $\alpha := 
\begin{bsmallmatrix}
x_1&x_2\\
x_3&x_1
\end{bsmallmatrix}$,
these concepts do not work anymore or they describe fundamentally different situations. For instance, the basic operation of substituting a single symbol in a word by another word cannot that easily be extended to the two-dimensional case. For example, the replacements $x_1 \mapsto \begin{bsmallmatrix}\ta&\ta\end{bsmallmatrix}, x_2 \mapsto \begin{bsmallmatrix}\tc\\\tc\end{bsmallmatrix}$ and $x_3 \mapsto \begin{bsmallmatrix}\tb\end{bsmallmatrix}$ may turn $\alpha$ into one of the following objects,
\begin{equation*}
\begin{bsmallmatrix}
&&\tc&\\
\ta&\ta&\tc&\\
&\tb&\ta&\ta
\end{bsmallmatrix} ,
\begin{bsmallmatrix}
\ta&\ta&\tc\\
&&\tc\\
\tb&\ta&\ta
\end{bsmallmatrix}, 
\begin{bsmallmatrix}
&&\tc\\
\ta&\ta&\tc\\
\tb&\ta&\ta
\end{bsmallmatrix},
\begin{bsmallmatrix}
&&&\tc\\
\ta&\ta&&\tc\\
&\tb&\ta&\ta
\end{bsmallmatrix},
\end{equation*}
which are not two-dimensional words, since they all contain holes or are not of rectangular shape and, most importantly, are not uniquely defined. On the other hand, it is straightforward to generalise the concept of a morphism to the two-dimensional case:

\begin{definition}\label{homoDef}
A mapping $h: \Sigma_1^{++} \rightarrow \Sigma_2^{++}$ is a two-dimensional morphism if it satisfies $h(V \varobar W) = h(V) \varobar h(W)$ and $h(V \varominus W)=h(V) \varominus h(W)$ for all $V, W \in \Sigma_1^{++}$. 
\end{definition}

Hence, we may say that a two-dimensional word $W$ satisfies a two-dimensional pattern $\alpha$ if and only if there exists a two-dimensional morphism which maps $\alpha$ to $W$. Unfortunately, this definition seems to be too strong as demonstrated by the following example. From an intuitive point of view, the two-dimensional word
$W := \begin{bsmallmatrix}
\ta& \ta& \tb& \ta& \ta& \tb\\
\ta& \ta& \tb& \ta& \ta& \tb\\
\tc& \tc& \tc&\tc&\tc&\tc
\end{bsmallmatrix}$
should satisfy the two-dimensional pattern
$\alpha := \begin{bsmallmatrix}
x_1 & x_1\\
x_2 & x_2
\end{bsmallmatrix}$,
but there is no two-dimensional morphism mapping $\alpha$ to $W$. This is due to the fact that, as pointed out by the following proposition (which has also been mentioned by Siromoney~et~al. in \cite{sir:pic}), a two-dimensional morphism is a mapping with a surprisingly strong condition.

\begin{proposition}\label{morphismProp}
Let $\Sigma_1 := \{a_1, a_2, \ldots, a_k\}$ and $\Sigma_2$ be alphabets. If a mapping $h : \Sigma_1^{++} \rightarrow \Sigma_2^{++}$ is a two-dimensional morphism, then
\begin{equation*}
|h(\begin{bsmallmatrix} a_1 \end{bsmallmatrix})|_c = |h(\begin{bsmallmatrix} a_2 \end{bsmallmatrix})|_c = \ldots = |h(\begin{bsmallmatrix} a_k \end{bsmallmatrix})|_c \text{ and } |h(\begin{bsmallmatrix} a_1 \end{bsmallmatrix})|_r = |h(\begin{bsmallmatrix} a_2 \end{bsmallmatrix})|_r = \ldots = |h(\begin{bsmallmatrix} a_k \end{bsmallmatrix})|_r\,.
\end{equation*}
\end{proposition}

\begin{proof}
We prove the statement of the proposition by contraposition. To this end, we assume that, for some $i, j$, $1 \leq i < j \leq k$, $|h(\begin{bsmallmatrix} a_i \end{bsmallmatrix})|_r \neq |h(\begin{bsmallmatrix} a_j \end{bsmallmatrix})|_r$, which implies $h(\begin{bsmallmatrix} a_i \end{bsmallmatrix}) \varobar h(\begin{bsmallmatrix} a_j \end{bsmallmatrix}) = \undef$. Hence, since $h(\begin{bsmallmatrix} a_i \end{bsmallmatrix} \varobar \begin{bsmallmatrix} a_j \end{bsmallmatrix}) = h(\begin{bsmallmatrix} a_i & a_j \end{bsmallmatrix}) \in \Sigma_2^{++}$, we can conclude that $h(\begin{bsmallmatrix} a_i \end{bsmallmatrix} \varobar \begin{bsmallmatrix} a_j \end{bsmallmatrix}) \neq h(\begin{bsmallmatrix} a_i \end{bsmallmatrix}) \varobar h(\begin{bsmallmatrix} a_j \end{bsmallmatrix})$, which contradicts the morphism property. Similarly, if $|h(\begin{bsmallmatrix} a_i \end{bsmallmatrix})|_c \neq |h(\begin{bsmallmatrix} a_j \end{bsmallmatrix})|_c$, then $h(\begin{bsmallmatrix} a_i \end{bsmallmatrix} \varominus 
\begin{bsmallmatrix} a_j \end{bsmallmatrix}) \neq h(\begin{bsmallmatrix} a_i \end{bsmallmatrix}) \varominus h(\begin{bsmallmatrix} a_j \end{bsmallmatrix})$.
\end{proof}

Similarly as in the string case, homomorphisms $h: \Sigma_1^{++} \rightarrow \Sigma_2^{++}$ are uniquely defined by giving the images $h(\Sigma_1)$.
If in particular $h(\Sigma_1)\subseteq\Sigma_2$, we term the resulting
morphism a \emph{letter-to-letter morphism}, while in the even more restricted
case when the restriction $h_{\Sigma_1}$ of $h$ to $\Sigma_1$ yields
a surjective mapping $h_{\Sigma_1}:\Sigma_1\to\Sigma_2$, $h$ is referred to as
a \emph{projection}.
We can conclude that the existence of a two-dimensional morphism seems to be a reasonable sufficient criterion for the situation that a two-dimensional word satisfies a two-dimensional pattern, but not a necessary one. \par
In fact, it turns out that characteristic factorisations provide the most promising approach to formalise how a two-dimensional word satisfies a two-dimensional pattern. Recall the example pattern 
$\alpha = \begin{bsmallmatrix}
x_1 & x_1\\
x_2 & x_2
\end{bsmallmatrix}$
from above. Since $\alpha = (\begin{bsmallmatrix} x_1 \end{bsmallmatrix} \varobar \begin{bsmallmatrix} x_1 \end{bsmallmatrix}) \varominus (\begin{bsmallmatrix} x_2 \end{bsmallmatrix} \varobar \begin{bsmallmatrix} x_2 \end{bsmallmatrix})$, a characteristic factorisation of a two-dimensional word $U$ for $\alpha$ is a factorisation of the form $U = (V_1 \varobar V_1) \varominus (V_2 \varobar V_2)$. On the other hand, since $\alpha = (\begin{bsmallmatrix} x_1 \end{bsmallmatrix} \varominus \begin{bsmallmatrix} x_2 \end{bsmallmatrix}) \varobar (\begin{bsmallmatrix} x_1 \end{bsmallmatrix} \varominus \begin{bsmallmatrix} x_2 \end{bsmallmatrix})$, we could as well regard a factorisation $U = (V_1 \varominus V_2) \varobar (V_1 \varominus V_2)$ as characteristic for $\alpha$. For the sake of convenience, we say that the former factorisation is of \emph{column-row type} and the latter one is of \emph{row-column type}. Obviously, the two-dimensional word $W$ from above has a characteristic factorisation of column-row 
type and a characteristic factorisation of row-column type (with respect to $\alpha$): 
\begin{align*}
&W = (V_1 \varobar V_1) \varominus (V_2 \varobar V_2) = (\begin{bsmallmatrix} \ta&\ta&\tb \\ \ta&\ta&\tb \end{bsmallmatrix} \varobar \begin{bsmallmatrix} \ta&\ta&\tb \\ \ta&\ta&\tb \end{bsmallmatrix}) \varominus (\begin{bsmallmatrix} \tc&\tc&\tc \end{bsmallmatrix} \varobar \begin{bsmallmatrix} \tc&\tc&\tc \end{bsmallmatrix}) = 
\begin{bsmallmatrix}
\ta& \ta& \tb& \ta& \ta& \tb\\
\ta& \ta& \tb& \ta& \ta& \tb\\
\tc& \tc& \tc&\tc&\tc&\tc
\end{bsmallmatrix}\,,\\
&W = (V_1 \varominus V_2) \varobar (V_1 \varominus V_2) = (\begin{bsmallmatrix} \ta&\ta&\tb \\ \ta&\ta&\tb \end{bsmallmatrix} \varominus \begin{bsmallmatrix} \tc&\tc&\tc \end{bsmallmatrix}) \varobar (\begin{bsmallmatrix} \ta&\ta&\tb \\ \ta&\ta&\tb \end{bsmallmatrix} \varominus \begin{bsmallmatrix} \tc&\tc&\tc \end{bsmallmatrix}) = 
\begin{bsmallmatrix}
\ta& \ta& \tb& \ta& \ta& \tb\\
\ta& \ta& \tb& \ta& \ta& \tb\\
\tc& \tc& \tc&\tc&\tc&\tc
\end{bsmallmatrix}\,.
\end{align*}
As a matter of fact, for every two-dimensional word $U$ there exists a characteristic factorisation for $\alpha = \begin{bsmallmatrix}
x_1 & x_1\\
x_2 & x_2
\end{bsmallmatrix}$ of column-row type if and only if there exists a characteristic factorisation for $\alpha$ of row-column type. However, this is a particularity of $\alpha$ and, e.\,g., for 
$\alpha' = \begin{bsmallmatrix}
x_1 & x_2 & x_3\\
x_2 & x_3 & x_1
\end{bsmallmatrix}$ and 
$W' := \begin{bsmallmatrix}
\ta& \ta& \ta& \tb& \tc\\
\tb& \tc&\ta& \ta& \ta\\
\end{bsmallmatrix}$,
there exists a characteristic factorisation of column-row type $W' = (V_1 \varobar V_2 \varobar V_3) \varominus (V_2 \varobar V_3 \varobar V_1) = (\begin{bsmallmatrix} \ta& \ta& \ta \end{bsmallmatrix} \varobar \begin{bsmallmatrix} \tb \end{bsmallmatrix} \varobar \begin{bsmallmatrix} \tc \end{bsmallmatrix}) \varominus (\begin{bsmallmatrix} \tb \end{bsmallmatrix} \varobar \begin{bsmallmatrix} \tc \end{bsmallmatrix} \varobar \begin{bsmallmatrix} \ta& \ta& \ta \end{bsmallmatrix})$, but no characteristic factorisation of row-column type. Furthermore, the column-row factorisation of $W'$ is somewhat at odds with our intuitive understanding of what it means that a two-dimensional word satisfies a two-dimensional pattern. This is due to the fact that factorising $W'$ into $(\begin{bsmallmatrix} \ta& \ta& \ta \end{bsmallmatrix} \varobar \begin{bsmallmatrix} \tb \end{bsmallmatrix} \varobar \begin{bsmallmatrix} \tc \end{bsmallmatrix}) \varominus (\begin{bsmallmatrix} \tb \end{bsmallmatrix} \varobar \begin{bsmallmatrix} 
\tc \end{bsmallmatrix} \varobar \begin{bsmallmatrix} \ta& \ta& \ta \end{bsmallmatrix})$ means that we associate the two-dimensional factors $\begin{bsmallmatrix} \ta& \ta& \ta \end{bsmallmatrix}$, $\begin{bsmallmatrix} \tb \end{bsmallmatrix}$ and $\begin{bsmallmatrix} \tc \end{bsmallmatrix}$ with the variables $x_1$, $x_2$ and $x_3$, respectively, but in the pattern $\alpha'$ the vertical neighbourship relation between the occurrence of $x_2$ in the first row and the occurrence of $x_3$ in the second row is not preserved in $W'$ with respect to the corresponding two-dimensional factors $\begin{bsmallmatrix} \tb \end{bsmallmatrix}$ and $\begin{bsmallmatrix} \tc \end{bsmallmatrix}$. More precisely, while a column-row factorisation preserves the horizontal neighbourship relation of the variables, it may violate their vertical neighbourship relation, where for row-column factorisations it is the other way around. Consequently, if we want both the vertical as well as the horizontal neighbourship relation to be 
preserved, we should require that the two-dimensional word $U$ can be disassembled into two-dimensional factors that induce both a column-row as well as a row-column factorisation. More precisely, we say that $U$ satisfies 
$\alpha' = \begin{bsmallmatrix}
x_1 & x_2 & x_3\\
x_2 & x_3 & x_1
\end{bsmallmatrix}$
if and only if there exist two-dimensional words $V_1, V_2$ and $V_3$, such that $U = (V_1 \varobar V_2 \varobar V_3) \varominus (V_2 \varobar V_3 \varobar V_1) = (V_1 \varominus V_2) \varobar (V_2 \varominus V_3) \varobar (V_3 \varominus V_1)$, which we call a \emph{proper} characteristic factorisation of $U$. \par
We are now ready to formalise the ideas developed so far and we can finally give a sound definition of two-dimensional pattern languages. Although we consider the class of two-dimensional pattern languages that results from the proper characteristic factorisations as the natural two-dimensional counterpart of the class of one-dimensional pattern languages, we shall also define the other classes of two-dimensional pattern languages which were sketched above. \par
For the definition of two-dimensional patterns, we use the same set of variables $X$ that has already been used in the definition of one-dimensional pattern languages. An \emph{array pattern} is a non-empty two-dimensional word over $X$ and a \emph{terminal array} is a non-empty two-dimensional word over a \emph{terminal} alphabet $\Sigma$. If it is clear from the context that we are concerned with array patterns and terminal arrays, then we simply say pattern and array, respectively. 
Any mapping $h : X \rightarrow \Sigma^{++}$ is called a \emph{substitution}. For any substitution $h$, by $h_{\varobar, \varominus}$, we denote the mapping $X^{++} \rightarrow \Sigma^{++}$ defined in the following way. For any $\alpha := [y_{i, j}]_{m, n} \in X^{++}$, we define 
\begin{align*}
h_{\varobar, \varominus}(\alpha) :=\:&(h(y_{1,1}) \varobar h(y_{1,2}) \varobar \ldots \varobar h(y_{1,n})) \varominus\\
&(h(y_{2,1}) \varobar h(y_{2,2}) \varobar \ldots \varobar h(y_{2,n})) \varominus \ldots \varominus\\
&(h(y_{m,1}) \varobar h(y_{m,2}) \varobar \ldots \varobar h(y_{m,n}))\,.
\end{align*}
Similarly, $h_{\varominus, \varobar} : X^{++} \rightarrow \Sigma^{++}$ is defined by 
\begin{align*}
h_{\varominus, \varobar}(\alpha) :=\:&(h(y_{1,1}) \varominus h(y_{2,1}) \varominus \ldots \varominus h(y_{m,1})) \varobar\\
&(h(y_{1,2}) \varominus h(y_{2,2}) \varominus \ldots \varominus h(y_{m,2})) \varobar \ldots \varobar\\
&(h(y_{1,n}) \varominus h(y_{2,n}) \varominus \ldots \varominus h(y_{m,n}))\,.
\end{align*}
Intuitively speaking, both mappings $h_{\varobar, \varominus}$ and $h_{\varobar, \varominus}$, when applied to an array pattern $\alpha$, first substitute every variable occurrence of $\alpha$ by a terminal array according to the substitution $h$ and then these $m \times n$ individual terminal arrays are assembled to one terminal array by either first column-concatenating all the $n$ terminal arrays in every individual row and then row-concatenating the resulting $m$ terminal arrays, or by first row-concatenating all the $m$ terminal arrays in every individual column and then column-concatenating the resulting $n$ terminal arrays.\par
Let $\alpha \in X^{++}$, $W \in \Sigma^{++}$ and let $h : X \rightarrow \Sigma^{++}$. The array $W$ is a (1) \emph{column-row image} of $\alpha$ (\emph{with respect to $h$}), (2) a \emph{row-column image} of $\alpha$ (\emph{with respect to $h$}) or (3) a \emph{proper image} of $\alpha$ (\emph{with respect to $h$}) if and only if (1) $h_{\varobar, \varominus}(\alpha) = W$, (2) $h_{\varominus, \varobar}(\alpha) = W$ or (3) $h_{\varobar, \varominus}(\alpha) = h_{\varominus, \varobar}(\alpha) = W$, respectively. The mapping $h$ is called a \emph{column-row substitution for $\alpha$ and $W$}, a \emph{row-column substitution for $\alpha$ and $W$} or a \emph{proper substitution for $\alpha$ and $W$}, respectively. We say that $W$ is a column-row, a row-column or a proper image of $\alpha$ if there exists a column-row, a row-column or a proper substitution, respectively, for $\alpha$ and $W$. \par
A nice and intuitive way to interpret the different kinds of images of array 
patterns is to imagine a grid to be placed over the terminal array. 
The vertical lines of the grid represent a column concatenation and the 
horizontal lines of the grid represent a row concatenation of the corresponding factorisation. 
This means that every rectangular area of the grid corresponds to an occurrence 
of a variable $x$ in the array pattern or, more precisely, 
to the array $h(x)$ substituted for $x$. The fact that an array satisfies a pattern 
is then represented by the situation that each two rectangular areas of the 
grid that correspond to occurrences of the same variable must have identical 
content. In Figure~\ref{fig-illustration}, an example for each a morphic image, 
a proper image, a column-row image and a row-column image of a $5 \times 4$ 
pattern is represented in this illustrative way.

\begin{figure}[tb]
\begin{center}
\begin{tikzpicture}[scale=.4]
\coordinate (coord0) at (0,0);
\coordinate (xshift) at (8.5,0);
\coordinate (width) at (7.5,0);
\coordinate (height) at (0,12.5);

\draw ($(0,0)$) -> ($(0,0) + (width)$);
\draw ($(0,0) + (width)$) -> ($(0,0) + (width) - (height)$);
\draw ($(0,0) + (width) - (height)$) -> ($(0,0) - (height)$);
\draw ($(0,0) - (height)$) -> ($(0,0)$);

\draw ($(0,0) + (1.875,0)$) -> ($(0,0) + (1.875,0) - (height)$);
\draw ($(0,0) + (3.75,0)$) -> ($(0,0) + (3.75,0) - (height)$);
\draw ($(0,0) + (5.625,0)$) -> ($(0,0) + (5.625,0) - (height)$);

\draw ($(0,0) - (0, 2.5)$) -> ($(0,0) - (0, 2.5) + (width)$);
\draw ($(0,0) - (0, 5)$) -> ($(0,0) - (0, 5) + (width)$);
\draw ($(0,0) - (0, 7.5)$) -> ($(0,0) - (0, 7.5) + (width)$);
\draw ($(0,0) - (0, 10)$) -> ($(0,0) - (0, 10) + (width)$);


\draw ($(0,0) + (xshift)$) -> ($(0,0) + (width) + (xshift)$);
\draw ($(0,0) + (width) + (xshift)$) -> ($(0,0) + (width) - (height) + (xshift)$);
\draw ($(0,0) + (width) - (height) + (xshift)$) -> ($(0,0) - (height) + (xshift)$);
\draw ($(0,0) - (height) + (xshift)$) -> ($(0,0) + (xshift)$);

\draw ($(0,0) + (2,0) + (xshift)$) -> ($(0,0) + (2,0) - (height) + (xshift)$);
\draw ($(0,0) + (3,0) + (xshift)$) -> ($(0,0) + (3,0) - (height) + (xshift)$);
\draw ($(0,0) + (4.5,0) + (xshift)$) -> ($(0,0) + (4.5,0) - (height) + (xshift)$);

\draw ($(0,0) - (0, 2)+ (xshift)$) -> ($(0,0) - (0, 2) + (width)+ (xshift)$);
\draw ($(0,0) - (0, 3.5)+ (xshift)$) -> ($(0,0) - (0, 3.5) + (width)+ (xshift)$);
\draw ($(0,0) - (0, 8)+ (xshift)$) -> ($(0,0) - (0, 8) + (width)+ (xshift)$);
\draw ($(0,0) - (0, 11)+ (xshift)$) -> ($(0,0) - (0, 11) + (width)+ (xshift)$);


\draw ($(0,0) + 2*(xshift)$) -> ($(0,0) + (width) + 2*(xshift)$);
\draw ($(0,0) + (width) + 2*(xshift)$) -> ($(0,0) + (width) - (height) + 2*(xshift)$);
\draw ($(0,0) + (width) - (height) + 2*(xshift)$) -> ($(0,0) - (height) + 2*(xshift)$);
\draw ($(0,0) - (height) + 2*(xshift)$) -> ($(0,0) + 2*(xshift)$);

\draw ($(0,0) + (2.5,0) + 2*(xshift)$) -> ($(0,0) + (2.5,0) - (0, 3) + 2*(xshift)$);
\draw ($(0,0) + (5,0) + 2*(xshift)$) -> ($(0,0) + (5,0) - (0, 3) + 2*(xshift)$);
\draw ($(0,0) + (5.5,0) + 2*(xshift)$) -> ($(0,0) + (5.5,0) - (0, 3) + 2*(xshift)$);

\draw ($(0,0) + (1,0) - (0, 3) + 2*(xshift)$) -> ($(0,0) + (1,0) - (0, 5) + 2*(xshift)$);
\draw ($(0,0) + (3.5,0) - (0, 3) + 2*(xshift)$) -> ($(0,0) + (3.5,0) - (0, 5) + 2*(xshift)$);
\draw ($(0,0) + (6,0) - (0, 3) + 2*(xshift)$) -> ($(0,0) + (6,0) - (0, 5) + 2*(xshift)$);

\draw ($(0,0) + (2,0) - (0, 5) + 2*(xshift)$) -> ($(0,0) + (2,0) - (0, 6.5) + 2*(xshift)$);
\draw ($(0,0) + (3.5,0) - (0, 5) + 2*(xshift)$) -> ($(0,0) + (3.5,0) - (0, 6.5) + 2*(xshift)$);
\draw ($(0,0) + (5,0) - (0, 5) + 2*(xshift)$) -> ($(0,0) + (5,0) - (0, 6.5) + 2*(xshift)$);

\draw ($(0,0) + (2.5,0) - (0, 6.5) + 2*(xshift)$) -> ($(0,0) + (2.5,0) - (0, 9.5) + 2*(xshift)$);
\draw ($(0,0) + (5,0) - (0, 6.5) + 2*(xshift)$) -> ($(0,0) + (5,0) - (0, 9.5) + 2*(xshift)$);
\draw ($(0,0) + (6.5,0) - (0, 6.5) + 2*(xshift)$) -> ($(0,0) + (6.5,0) - (0, 9.5) + 2*(xshift)$);

\draw ($(0,0) + (1,0) - (0, 9.5) + 2*(xshift)$) -> ($(0,0) + (1,0) - (height) + 2*(xshift)$);
\draw ($(0,0) + (4,0) - (0, 9.5) + 2*(xshift)$) -> ($(0,0) + (4,0) - (height) + 2*(xshift)$);
\draw ($(0,0) + (6,0) - (0, 9.5) + 2*(xshift)$) -> ($(0,0) + (6,0) - (height) + 2*(xshift)$);

\draw ($(0,0) - (0, 3)+ 2*(xshift)$) -> ($(0,0) - (0, 3) + (width)+ 2*(xshift)$);
\draw ($(0,0) - (0, 5)+ 2*(xshift)$) -> ($(0,0) - (0, 5) + (width)+ 2*(xshift)$);
\draw ($(0,0) - (0, 6.5)+ 2*(xshift)$) -> ($(0,0) - (0, 6.5) + (width)+ 2*(xshift)$);
\draw ($(0,0) - (0, 9.5)+ 2*(xshift)$) -> ($(0,0) - (0, 9.5) + (width)+ 2*(xshift)$);


\draw ($(0,0) + 3*(xshift)$) -> ($(0,0) + (width) + 3*(xshift)$);
\draw ($(0,0) + (width) + 3*(xshift)$) -> ($(0,0) + (width) - (height) + 3*(xshift)$);
\draw ($(0,0) + (width) - (height) + 3*(xshift)$) -> ($(0,0) - (height) + 3*(xshift)$);
\draw ($(0,0) - (height) + 3*(xshift)$) -> ($(0,0) + 3*(xshift)$);

\draw ($(0,0) + (2,0) + 3*(xshift)$) -> ($(0,0) + (2,0) - (height) + 3*(xshift)$);
\draw ($(0,0) + (3,0) + 3*(xshift)$) -> ($(0,0) + (3,0) - (height) + 3*(xshift)$);
\draw ($(0,0) + (6,0) + 3*(xshift)$) -> ($(0,0) + (6,0) - (height) + 3*(xshift)$);

\draw ($(0,0) - (0, 2)+ 3*(xshift)$) -> ($(0,0) - (0, 2) + (2,0) + 3*(xshift)$);
\draw ($(0,0) - (0, 1.5) + (2,0) + 3*(xshift)$) -> ($(0,0) - (0, 1.5) + (3,0) + 3*(xshift)$);
\draw ($(0,0) - (0, 2.5) + (3,0) + 3*(xshift)$) -> ($(0,0) - (0, 2.5) + (6,0) + 3*(xshift)$);
\draw ($(0,0) - (0, 4) + (6,0) + 3*(xshift)$) -> ($(0,0) - (0, 4) + (7.5,0) + 3*(xshift)$);

\draw ($(0,0) - (0, 4.5)+ 3*(xshift)$) -> ($(0,0) - (0, 4.5) + (2,0)+ 3*(xshift)$);
\draw ($(0,0) - (0, 6)+ (2,0) + 3*(xshift)$) -> ($(0,0) - (0, 6) + (3,0)+ 3*(xshift)$);
\draw ($(0,0) - (0, 7)+ (3,0) + 3*(xshift)$) -> ($(0,0) - (0, 7) + (6,0)+ 3*(xshift)$);
\draw ($(0,0) - (0, 5)+ (6,0) + 3*(xshift)$) -> ($(0,0) - (0, 5) + (7.5,0)+ 3*(xshift)$);

\draw ($(0,0) - (0, 8)+ 3*(xshift)$) -> ($(0,0) - (0, 8) + (2,0)+ 3*(xshift)$);
\draw ($(0,0) - (0, 8) + (2,0) + 3*(xshift)$) -> ($(0,0) - (0, 8) + (3,0)+ 3*(xshift)$);
\draw ($(0,0) - (0, 9) + (3,0) + 3*(xshift)$) -> ($(0,0) - (0, 9) + (6,0)+ 3*(xshift)$);
\draw ($(0,0) - (0, 8.5) + (6,0) + 3*(xshift)$) -> ($(0,0) - (0, 8.5) + (7.5,0)+ 3*(xshift)$);

\draw ($(0,0) - (0, 12)+ 3*(xshift)$) -> ($(0,0) - (0, 12) + (2,0)+ 3*(xshift)$);
\draw ($(0,0) - (0, 11)+ (2,0) + 3*(xshift)$) -> ($(0,0) - (0, 11) + (3,0)+ 3*(xshift)$);
\draw ($(0,0) - (0, 11.5)+ (3,0) + 3*(xshift)$) -> ($(0,0) - (0, 11.5) + (6,0)+ 3*(xshift)$);
\draw ($(0,0) - (0, 10)+ (6,0) + 3*(xshift)$) -> ($(0,0) - (0, 10) + (7.5,0)+ 3*(xshift)$);


\node[anchor=center] at ($(0, 0) + (3.75, 0) - (height) - (0,0.8)$) {morphic image};
\node[anchor=center] at ($(0, 0) + (3.75, 0) - (height) - (0,0.8) + (xshift)$) {proper image};
\node[anchor=center] at ($(0, 0) + (3.75, 0) - (height) - (0,0.8) + 2*(xshift)$) {column-row image};
\node[anchor=center] at ($(0, 0) + (3.75, 0) - (height) - (0,0.8) + 3*(xshift)$) {row-column image};

\end{tikzpicture}
\end{center}
\caption{Illustrating possible image partitions.}
\label{fig-illustration}
\end{figure}
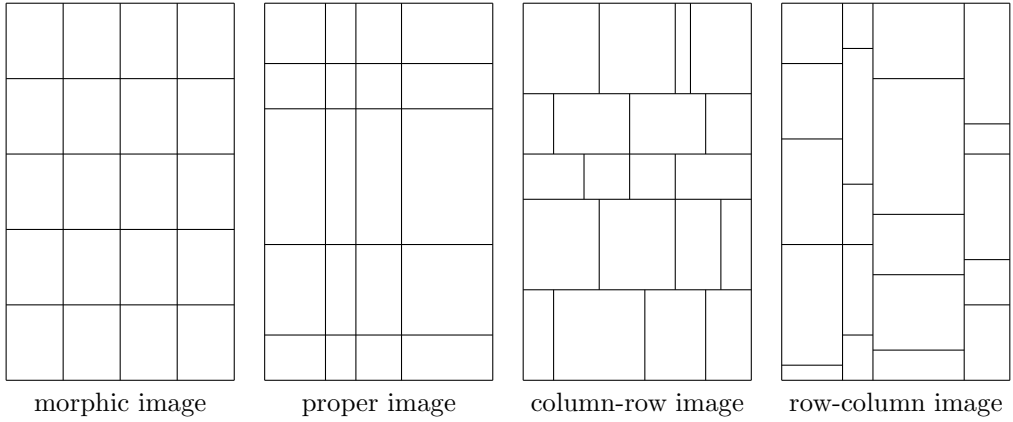

Alternatively, we can interpret the property that a terminal array $W$ is a certain type of image of an array pattern as a tiling of $W$. More precisely, $W$ satisfies a given array pattern $\alpha$ with $n$ different variables if and only if $n$ tiles can be allocated to the $n$ variables of $\alpha$ such that combining the tiles as indicated by the structure of $\alpha$ yields $W$. The grids depicted in Figure~\ref{fig-illustration} then illustrate the structure of such a tiling.
The definitions of the corresponding classes of pattern languages are now straightforward:
\begin{definition}
Let $\alpha \in X^{++}$ be an array pattern. We define the following variants of two-dimensional pattern languages:
\begin{itemize}
\item $L_{\Sigma, h}(\alpha) := \{W \in \Sigma^{**} \mid W \text{ is a morphic image of } \alpha\}$,
\item $L_{\Sigma, p}(\alpha) := \{W \in \Sigma^{**} \mid W \text{ is a proper image of } \alpha\}$,
\item $L_{\Sigma, r}(\alpha) := \{W \in \Sigma^{**} \mid W \text{ is a column-row image of } \alpha\}$,
\item $L_{\Sigma, c}(\alpha) := \{W \in \Sigma^{**} \mid W \text{ is a row-column image of } \alpha\}$,
\item $L_{\Sigma, rc}(\alpha) := L_{\Sigma, r}(\alpha) \cup L_{\Sigma, c}(\alpha)$.
\end{itemize}
For a pattern $\alpha$, we also denote the above languages by \emph{$Z$ pattern language of $\alpha$}, where $Z \in \{h, p, r, c, rc\}$.
For every $x \in \{r, c, rc, p, h\}$, we define $\mathcal{L}_{\Sigma, x} := \{L_{\Sigma, x}(\alpha) \mid \alpha \in X^{++}\}$ and $\mathcal{L}_{x} := \{\mathcal{L}_{\Sigma, x} \mid \Sigma \text{ is some alphabet}\}$.
\end{definition}

Since, for a fixed array pattern $\alpha$, every morphic image is a proper image and every proper image is a row-column image as well as a column-row image, the following subset relations between the different types of pattern languages hold (in the following diagram, an arrow denotes a subset relation):\par
\begin{tikzpicture}
[every path/.style={>=latex},every node/.style={text centered}]
  \node            (0) at (0,0) {$\:$};
  \node            (1) at (2,0) {$L_{\Sigma, h}(\alpha)$};
  \node            (2) at (5,0) {$L_{\Sigma, p}(\alpha)$}; 
  \node            (3a) at (8,-0.3) {$L_{\Sigma, c}(\alpha)$};
  \node            (3b) at (8,0.3) {$L_{\Sigma, r}(\alpha)$};
  \node            (4) at (11,0) {$L_{\Sigma, rc}(\alpha)$};

  \draw[->] (1) -- (2);
  \draw[->] (2) -- (3a);
  \draw[->] (2) -- (3b);
  \draw[->] (3a) -- (4);
  \draw[->] (3b) -- (4);
\end{tikzpicture}
\begin{remark}
As indicated in the introductory part of this section, we consider the class of $p$ pattern languages as the most natural class of two-dimensional pattern languages. Another observation that supports this claim is that the $p$ pattern languages are compatible, in a certain sense, to the one-dimensional pattern languages. More precisely, for a \emph{one-dimensional} (i.\,e., $1 \times n$) array pattern $\alpha$ the set $L_{\Sigma, p}(\alpha) \cap \{W \in \Sigma^{++} \mid |W|_r = 1\}$ coincides with the one-dimensional pattern language of $\alpha$. This does not hold for the $h$ pattern languages (since in the one-dimensional case the words variables are mapped to can differ in length), but holds for the $r$, $c$ and $rc$ pattern languages. However, as pointed out above, the $r$, $c$ and $rc$ pattern language of a given pattern $\alpha$ may contain arrays that, from an intuitive point of view, do not satisfy $\alpha$.
\end{remark}

\section{General Observations}

In this section, we state some general lemmas about two-dimensional morphisms and array pattern languages, which shall be important for proving the further results presented in this paper. First, we refine Proposition~\ref{morphismProp}, by giving a convenient characterisation for the morphism property for mappings on arrays. To this end, we define a substitution $h$ to be \emph{$(m, n)$-uniform} if, for every $x \in X$, $|h(x)|_{r} = m$ and $|h(x)|_{c} = n$ and a substitution is \emph{uniform} if it is $(m, n)$-uniform, for some $m, n \in \mathbb{N}$.

\begin{lemma}\label{lem-hom-charac}
A mapping $h:\Sigma^{**} \to \Gamma^{**}$ is a two-dimensional morphism if and only if $h = g_{\varominus, \varobar} = g_{\varobar, \varominus}$, where $g : \Sigma \to \Gamma^{**}$ is a uniform substitution.
\end{lemma}

\begin{proof}
We first observe that if $g$ is uniform, then $g_{\varominus, \varobar} = g_{\varobar, \varominus}$ obviously holds (so it is sufficient to prove the statement of the lemma only for one of these two mappings). Furthermore, if $g$ is uniform, then, for every $U, V \in \Sigma^{**}$, $g_{\varobar, \varominus}(U \varominus V) = g_{\varobar, \varominus}(U) \varominus g_{\varobar, \varominus}(V)$ and $g_{\varobar, \varominus}(U \varobar V) = g_{\varobar, \varominus}(U) \varobar g_{\varobar, \varominus}(V)$, which proves the \emph{if} direction. In order to prove the \emph{only if} direction, we assume that $h$ is a two-dimensional morphism and we define a substitution $g$ by $g(b) := h(b)$, $b \in \Sigma$. Furthermore, let $\widehat{g} \in \{g_{\varominus, \varobar}, g_{\varobar, \varominus}\}$. We now show that $h$ equals $\widehat{g}$ by induction. By definition, for every $b \in \Sigma$, $h(b) = \widehat{g}(b)$. Now let $U, V \in \Sigma^{**}$ with $h(U) = g'(U)$ and $h(V) = g'(V)$. Then $h(U \varominus V) = h(U) \varominus h(V) = g'(U) \varominus g'(V) = g'(U \varominus V)$ and, analogously, $h(U \varobar V) = h(U) \varobar h(V) = g'(U) \varobar g'(V) = g'(U \varobar V)$. By induction, it follows that $h = \widehat{g}$. Consequently, we can conclude that $g$ is uniform. 
\end{proof}

The next lemma states that the composition of two two-dimensional morphisms is again a two-dimensional morphism.

\begin{lemma}\label{lem-hom-closure}
Let $h_1:\Gamma^{**}\to \Sigma^{**}$ and $h_2:\Sigma^{**}\to\Delta^{**}$ be two-dimensional morphisms.
Then, the composition $h_{1,2}:=h_1\circ h_2:\Gamma^{**}\to\Delta^{**}$ is also a two-dimensional morphism.
 \end{lemma}

\begin{proof}
We first observe the following. If $g$ and $f$ are some uniform substitutions, then $h := g \circ f$ is a uniform substitution as well. Furthermore, it can be easily verified that $g_{\varominus, \varobar} \circ f_{\varominus, \varobar} = h_{\varominus, \varobar}$. With Lemma~\ref{lem-hom-charac}, this directly implies the statement of the lemma.
\end{proof}

It is intuitively clear that the structure of a pattern fully determines the corresponding pattern language and the actual names of the variables are irrelevant, e.\,g., the patterns $\begin{bsmallmatrix} x_1 & x_2 & x_1 \\ x_2 & x_3 & x_3 \end{bsmallmatrix}$ and $\begin{bsmallmatrix} x_7 & x_3 & x_7 \\ x_3 & x_5 & x_5 \end{bsmallmatrix}$ should be considered identical. In the following we formalise this intuition. Two array patterns $\alpha := [y_{i, j}]_{m,n}$ and $\beta := [z_{i, j}]_{m',n'}$ are \emph{equivalent up to a renaming}, denoted by $\alpha \sim \beta$, if and only if $m = m'$, $n = n'$ and, for every $i, j, i', j'$, $1 \leq i, i' \leq m$, $1 \leq j, j' \leq n$, $y_{i, j} = y_{i', j'}$ if and only if $z_{i, j} = z_{i', j'}$. 

\begin{lemma}\label{equalityLemma}
Let $z, z' \in \{h, p, r, c, rc\}$, let $\Sigma$ be an alphabet with $|\Sigma| \geq 2$ and let $\alpha, \beta \in X^{++}$. If $L_{\Sigma, z}(\alpha) = L_{\Sigma, z'}(\beta)$, then $\alpha \sim \beta$.
\end{lemma}

\begin{proof}
We assume that $L_{\Sigma, z}(\alpha) = L_{\Sigma, z'}(\beta)$ and note that this implies that $|\alpha|_c = |\beta|_c = m$ and $|\alpha|_r = |\beta|_r = n$. This is due to the fact that if $|\alpha|_c < |\beta|_c$ or $|\alpha|_r < |\beta|_r$, then the array obtained from $\alpha$ by replacing every variable by a single symbol $\ta \in \Sigma$ is in $L_{\Sigma, z}(\alpha)$, but not in $L_{\Sigma, z'}(\beta)$. We further assume that $\alpha \nsim \beta$, which implies that there are $i, j, i', j'$ with $1 \leq i, i' \leq m$ and $1 \leq j, j' \leq n$, such that $\alpha[i, j] = \alpha[i', j']$ and $\beta[i, j] \neq \beta[i', j']$ (or $\alpha[i, j] \neq \alpha[i', j']$ and $\beta[i, j] = \beta[i', j']$, for which an analogous argument applies). We now define a substitution $h : X \rightarrow \Sigma$ in the following way. For every $x \in X$, if $\beta[i, j] = x$, then $h(x) := \ta$ and if $\beta[i, j] \neq x$, then $h(x) := \tb$. We observe that, since $h_{\varominus, \varobar}(\beta)$ is a morphic image, a proper image, a row-column image and a column-row image of $\beta$, $h_{\varominus, \varobar}(\beta) \in L_{\Sigma, z'}(\beta)$. Furthermore, $h_{\varominus, \varobar}(\beta)[i, j] \neq h_{\varominus, \varobar}(\beta)[i', j']$. On the other hand, for every $W \in L_{\Sigma, z}(\alpha)$, with $|W|_c = |h_{\varominus, \varobar}(\beta)|_c$ and $|W|_r = |h_{\varominus, \varobar}(\beta)|_r$, $W[i, j] = W[i', j']$ is 
satisfied. Thus, $h_{\varominus, \varobar}(\beta) \notin L_{\Sigma, z}(\alpha)$, which implies that $L_{\Sigma, z}(\alpha) \neq L_{\Sigma, z'}(\beta)$.
\end{proof}

For every $z, z' \in \{h, p, r, c, rc\}$, $z \neq z'$, $\alpha \sim \beta$ does not necessarily imply $L_{\Sigma, z}(\alpha) = L_{\Sigma, z'}(\beta)$, as pointed out by, e.\,g., $L_{\Sigma, h}(\begin{bsmallmatrix} x_1 & x_2 \end{bsmallmatrix}) \neq L_{\Sigma, p}(\begin{bsmallmatrix} x_1 & x_2 \end{bsmallmatrix})$ or $L_{\Sigma, p}(\begin{bsmallmatrix} x_1 & x_2 \\ x_3 & x_1 \end{bsmallmatrix}) \neq L_{\Sigma, c}(\begin{bsmallmatrix} x_1 & x_2 \\ x_3 & x_1 \end{bsmallmatrix})$. On the other hand, since, for every $z, \in \{h, p, r, c, rc\}$, $\alpha \sim \beta$ obviously implies $L_{\Sigma, z}(\alpha) = L_{\Sigma, z}(\beta)$, two $z$ pattern languages are equivalent if and only if they are described by two patterns that are equivalent up to a renaming.\par
In the remainder of this work, we do not distinguish anymore between patterns that are equivalent up to a renaming, i.\,e., from now on we say that $\alpha$ and $\beta$ are equivalent, denoted by $\alpha = \beta$ for simplicity, if they are actually the same arrays or if they are equivalent up to a renaming.

\section{Comparison of Array Pattern Language Classes}

In this section, we provide a pairwise comparison of our different classes of array pattern languages and, furthermore, we compare them with the class of recognisable array languages, denoted by REC, which is one of the most prominent classes of array languages. For a detailed description of REC, the reader is referred to the survey \cite{gia:two} by Giammarresi~and~Restivo. Next, we show that, for every alphabet $\Sigma$ with $|\Sigma| \geq 2$, the language classes REC, $\mathcal{L}_{\Sigma, h}$, $\mathcal{L}_{\Sigma, p}$, $\mathcal{L}_{\Sigma, r}$, $\mathcal{L}_{\Sigma, c}$ and $\mathcal{L}_{\Sigma, rc}$ are pairwise incomparable. More precisely, for every $\mathcal{L}_1, \mathcal{L}_2 \in \{\text{REC}, \mathcal{L}_{\Sigma, h}, \mathcal{L}_{\Sigma, p}, \mathcal{L}_{\Sigma, r}, \mathcal{L}_{\Sigma, c}, \mathcal{L}_{\Sigma, rc}\}$ with $\mathcal{L}_1 \neq \mathcal{L}_2$, we show that $\mathcal{L}_1 \setminus \mathcal{L}_2 \neq \emptyset$, $\mathcal{L}_2 \setminus \mathcal{L}_1 \neq \emptyset$ and $\mathcal{L}_1 \cap \mathcal{L}_2 \neq \emptyset$. The non-emptiness of the pairwise intersections of these language classes can be easily seen:
\begin{proposition}
For every $z \in \{h,p,r,c,rc\}$, $L_{z, \Sigma}(\begin{bsmallmatrix} x_1 \end{bsmallmatrix}) = \Sigma^{++}$ and $\Sigma^{++} \in \text{REC}$.
\end{proposition}
It remains to find, for every $\mathcal{L}_1, \mathcal{L}_2 \in \{\text{REC}, \mathcal{L}_{\Sigma, h}, \mathcal{L}_{\Sigma, p}, \mathcal{L}_{\Sigma, r}, \mathcal{L}_{\Sigma, c}, \mathcal{L}_{\Sigma, rc}\}$, a separating language $L_1 \in \mathcal{L}_1 \setminus \mathcal{L}_2$ and a separating language $L_2 \in \mathcal{L}_2 \setminus \mathcal{L}_1$. We first present all these separating languages in a table and then we formally prove their separating property. In rows $2$ to $6$ of the following table, if a pattern $\alpha$ is the entry that corresponds to the row labeled by class $\mathcal{L}_{\Sigma, z}$ and the column labeled by class $\mathcal{L}_{\Sigma, z'}$, where $z, z' \in \{h, p, r, c, rc\}$, $z \neq z'$, then this means that $L \in \mathcal{L}_{\Sigma, z} \setminus \mathcal{L}_{\Sigma, z'}$. Row $1$, on the other hand, contains recognisable array languages that are not array pattern languages.

\begin{center}
  \begin{tabular}{| c !{\vrule width 1pt} c | c | c | c | c | c |}
    \hline
     & REC & \textbf{$\mathcal{L}_{\Sigma, h}$} & $\mathcal{L}_{\Sigma, p}$ & $\mathcal{L}_{\Sigma, r}$ & $\mathcal{L}_{\Sigma, c}$ & $\mathcal{L}_{\Sigma, rc}$ \\ \noalign{\hrule height 1pt}   
    REC & -- & $\{\begin{bsmallmatrix}  \ta \end{bsmallmatrix}\}$ & $\{\begin{bsmallmatrix}  \ta \end{bsmallmatrix}\}$ & $\{\begin{bsmallmatrix}  \ta \end{bsmallmatrix}\}$ & $\{\begin{bsmallmatrix}  \ta \end{bsmallmatrix}\}$ & $\{\begin{bsmallmatrix}  \ta \end{bsmallmatrix}\}$ \\ \hline 
    $\mathcal{L}_{\Sigma, h}$ & $\begin{bsmallmatrix} x_1 \\ x_1 \end{bsmallmatrix}$ & -- & $\begin{bsmallmatrix} x_1&x_2 \\ x_3&x_4 \\ \end{bsmallmatrix}$ & $\begin{bsmallmatrix} x_1&x_2 \\ x_3&x_4 \\ \end{bsmallmatrix}$ & $\begin{bsmallmatrix} x_1&x_2 \\ x_3&x_4 \\ \end{bsmallmatrix}$ & $\begin{bsmallmatrix} x_1&x_2 \\ x_3&x_4 \\ \end{bsmallmatrix}$ \\ \hline 
    $\mathcal{L}_{\Sigma, p}$ & $\begin{bsmallmatrix} x_1 \\ x_1 \end{bsmallmatrix}$ & $\begin{bsmallmatrix} x_1&x_2 \\ x_3&x_4 \\ \end{bsmallmatrix}$ & -- & $\begin{bsmallmatrix} x_1&x_2 \\ x_2&x_1 \end{bsmallmatrix}$ & $\begin{bsmallmatrix} x_1&x_2 \\ x_2&x_1 \end{bsmallmatrix}$ & $\begin{bsmallmatrix} x_1&x_2 \\ x_2&x_1 \end{bsmallmatrix}$ \\ \hline
    $\mathcal{L}_{\Sigma, r}$ & $\begin{bsmallmatrix} x_1 \\ x_1 \end{bsmallmatrix}$ & $\begin{bsmallmatrix} x_1&x_2 \\ x_3&x_4 \\ \end{bsmallmatrix}$ & $\begin{bsmallmatrix} x_1&x_2 \\ x_2&x_1 \end{bsmallmatrix}$ & -- & $\begin{bsmallmatrix} x_1&x_2 \\ x_2&x_1 \end{bsmallmatrix}$ & $\begin{bsmallmatrix} x_1&x_2 \\ x_2&x_1 \end{bsmallmatrix}$ \\ \hline 
    $\mathcal{L}_{\Sigma, c}$ & $\begin{bsmallmatrix} x_1 \\ x_1 \end{bsmallmatrix}$ & $\begin{bsmallmatrix} x_1&x_2 \\ x_3&x_4 \\ \end{bsmallmatrix}$ & $\begin{bsmallmatrix} x_1&x_2 \\ x_2&x_1 \end{bsmallmatrix}$ & $\begin{bsmallmatrix} x_1&x_2 \\ x_2&x_1 \end{bsmallmatrix}$  & -- & $\begin{bsmallmatrix} x_1&x_2 \\ x_2&x_1 \end{bsmallmatrix}$ \\ \hline 
    $\mathcal{L}_{\Sigma, rc}$ & $\begin{bsmallmatrix} x_1 \\ x_1 \end{bsmallmatrix}$ & $\begin{bsmallmatrix} x_1&x_2 \\ x_3&x_4 \\ \end{bsmallmatrix}$ & $\begin{bsmallmatrix} x_1&x_2 \\ x_2&x_1 \end{bsmallmatrix}$ & $\begin{bsmallmatrix} x_1&x_2 \\ x_2&x_1 \end{bsmallmatrix}$ & $\begin{bsmallmatrix} x_1&x_2 \\ x_2&x_1 \end{bsmallmatrix}$ & -- \\ \hline 
    \end{tabular}
\end{center}

\begin{lemma}\label{nonRECLang}
$L_{\Sigma, h}(\begin{bsmallmatrix} x_1 \\ x_1 \end{bsmallmatrix}) \notin \text{REC}$.
\end{lemma}

\begin{proof}
In this proof, we use the characterisation of REC by local array languages and projections (see Giammarresi~and~Restivo \cite{gia:two} and also the next section). Let $\alpha := \begin{bsmallmatrix} x_1 \\ x_1 \end{bsmallmatrix}$ and $L := L_{\Sigma, h}(\alpha) = \{W \varominus W \mid W \in \Sigma^{++}\}$. Suppose $L \in REC$. Then there is a local array language $L'$ over an alphabet $\Gamma$ so that $L$ is a projection of $L'$. For the sake of convenience, we define $r := |\Gamma|$ and $s := |\Sigma|$. For every $m, n \in \mathbb{N}$, let
\begin{equation*}
L_{m,n} := \{W \varominus W \mid W \in \Sigma^{++} \text{ and } |W|_r = m, |W|_c = n\} \subseteq L\,.
\end{equation*}
Obviously, $|L_{m,n}| = s^{mn}$. Now let $L'_{m,n}$ be the set of pictures in $L^\prime,$ whose projections are in $L_{m,n}$. In the arrays of $L^\prime_{m,n},$ there are at most $r^{2n}$ possibilities of how the $m^{\text{th}}$ and $(m+1)^{\text{th}}$ row can look like. For sufficiently large $m,$ $s^{mn} > r^{2n}.$ Thus, there exist two arrays $W_1 \varominus W_1$ and  $W_2 \varominus W_2,$ $W_1 \neq W_2$ in $L_{m,n}$ such that the corresponding arrays $W_1^\prime \varominus W_1^\prime$ and $W_2^\prime \varominus W_2^\prime$ in $L^\prime_{m,n}$ have the same $m^{\text{th}}$ row and the same $(m + 1)^{\text{th}}$ row. Hence, since $L'$ is a local array language, $W_1^\prime \varominus W_2^\prime, W_2^\prime \varominus W_1^\prime \in L^\prime_{m,n}$ and therefore $W_1 \varominus W_2, W_2 \varominus W_1 \in L_{m,n} \subseteq L$, which is a contradiction.
\end{proof}
 
It can be easily verified that, for every $z \in \{p,r,c,rc\}$, $L_{\Sigma, z}(\alpha) = L_{\Sigma, h}(\alpha)$, where $\alpha := \begin{bsmallmatrix} x_1 \\ x_1 \end{bsmallmatrix}$. Hence, for every $z \in \{h, p, r, c, rc\}$, $L_{\Sigma, z}(\alpha) \notin \text{REC}$, which implies the first column of the table. Furthermore, for every $z \in \{h,p,r,c,rc\}$, $\{\begin{bsmallmatrix}  \ta \end{bsmallmatrix}\} \notin L_{\Sigma, z}(\alpha)$, but $\{\begin{bsmallmatrix}  \ta \end{bsmallmatrix}\} \in \text{REC}$, which implies the first row of the table. \par
We point out that, by Lemma~\ref{equalityLemma}, for every $z, z' \in \{h,p,r,c,rc\}$, $z \neq z'$, if there exists a pattern $\beta$ with $L_{\Sigma, z}(\beta) \neq L_{\Sigma, z'}(\beta)$, then $L_{\Sigma, z}(\beta) \in \mathcal{L}_{\Sigma, z} \setminus \mathcal{L}_{\Sigma, z'}$ and $L_{\Sigma, z'}(\beta) \in \mathcal{L}_{\Sigma, z'} \setminus \mathcal{L}_{\Sigma, z}$. Consequently, in order to prove the remaining entries of the table, it is sufficient to identify, for every $z, z' \in \{h,p,r,c,rc\}$, $z \neq z'$, a pattern $\beta$ with $L_{\Sigma, z}(\beta) \neq L_{\Sigma, z'}(\beta)$, which is done by the following two lemmas.

\begin{lemma}\label{hIncompLemma}
For every $z \in \{p,c,r,cr\}$, $L_{\Sigma, h}(\begin{bsmallmatrix} x_1&x_2 \\ x_3&x_4 \\ \end{bsmallmatrix}) \neq L_{\Sigma, z}(\begin{bsmallmatrix} x_1&x_2 \\ x_3&x_4 \\ \end{bsmallmatrix})$.
\end{lemma}

\begin{proof}
Let $\beta := 
\begin{bsmallmatrix} 
x_1&x_2 \\ 
x_3&x_4 \\ 
\end{bsmallmatrix}$ and let $W :=
\begin{bsmallmatrix}
  \ta & \ta & \ta  \\
  \ta & \ta & \ta  \\
  \ta & \ta & \ta  \\
 \end{bsmallmatrix}
$. For every $z \in \{p,c,r,cr\}$, $W \in \mathcal{L}_{\Sigma, z}(\beta)$, since $g_{\varominus, \varobar}(\beta) = g_{\varobar, \varominus}(\beta) = W$, where $g$ is defined by
 \begin{equation*}
g(x_1) :=
\begin{bsmallmatrix}
  \ta & \ta \\
  \ta & \ta \\
 \end{bsmallmatrix},\:\:\\
g(x_2) :=
\begin{bsmallmatrix}
  \ta \\
  \ta \\
 \end{bsmallmatrix},\:\:\\
 g(x_3) := \begin{bsmallmatrix}
  \ta & \ta \\
 \end{bsmallmatrix},\:\:\\
 g(x_4) := \begin{bsmallmatrix}
  \ta \\
 \end{bsmallmatrix}\,.
\end{equation*}
By Lemma~\ref{lem-hom-charac}, it is obvious that there does not exist any morphism $h$ with $h(\beta) = W$. Thus, $W \notin L_{\Sigma,h}(\beta)$ and, for every $z \in \{p, c, r, cr\}$, $L_{\Sigma,h}(\beta) \neq L_{\Sigma,z}(\beta)$.
\end{proof}

\begin{lemma}\label{prcIncompLemma}
For every $z, z' \in \{p,c,r,cr\}$, $z \neq z'$, $L_{\Sigma, z}(\begin{bsmallmatrix} x_1&x_2 \\ x_2&x_1 \end{bsmallmatrix}) \neq L_{\Sigma, z'}(\begin{bsmallmatrix} x_1&x_2 \\ x_2&x_1 \\ \end{bsmallmatrix})$.
\end{lemma}

\begin{proof}
Let $\gamma := 
\begin{bsmallmatrix} 
x_1&x_2 \\ 
x_2&x_1 
\end{bsmallmatrix}$ and let 
$W_1 :=
\begin{bsmallmatrix}
  \ta & \ta\\
  \ta & \ta\\
  \ta & \ta\\
 \end{bsmallmatrix},
 W_2 :=
\begin{bsmallmatrix}
  \ta & \ta & \ta\\
  \ta & \ta & \ta\\
\end{bsmallmatrix}$.
We observe that $g_{\varominus, \varobar}(\gamma) = W_1$, $g'_{\varobar, \varominus}(\gamma) = W_2,$ where $g, g'$ are defined by 
 \begin{equation*}
g(x_1) :=
\begin{bsmallmatrix}
  \ta \\
 \end{bsmallmatrix},\,
g(x_2) :=
\begin{bsmallmatrix}
  \ta\\
  \ta\\
 \end{bsmallmatrix},\,
 g'(x_1) := \begin{bsmallmatrix}
  \ta\\
 \end{bsmallmatrix},\,
 g'(x_2) := \begin{bsmallmatrix}
  \ta & \ta \\
 \end{bsmallmatrix}\,.
\end{equation*}
Thus, $W_1 \in L_{\Sigma, c}(\gamma)$, $W_2 \in L_{\Sigma, r}(\gamma)$ and $W_1, W_2 \in L_{\Sigma, rc}(\gamma)$. On the other hand, $W_1, W_2 \notin L_{\Sigma, p}(\gamma)$, since every proper image of $\gamma$ must have an even number of columns and an even number of rows. Consequently, for every $z \in \{r,c, rc\}$, $L_{\Sigma, p}(\gamma) \neq L_{\Sigma, z}(\gamma)$. Similarly, $W_1 \notin L_{\Sigma, r}$ and $W_2 \notin L_{\Sigma, c}$, since every column-row image of $\gamma$ must have an even number of rows and every row-column image of $\gamma$ must have an even number of columns. This implies that, for every $z, z' \in \{c,r,cr\}$, $z \neq z'$, $L_{\Sigma, z}(\gamma) \neq L_{\Sigma, z'}(\gamma)$, which concludes the proof.
\end{proof}

\section{Closure Properties of Array Pattern Languages}

The research of closure properties of classes of formal languages is a classical topic in this area. However, the number of natural properties is richer in the case of arrays compared to the more conventional string case. Thus, in this section, we classify the operations that shall be investigated in this regard according to whether or not they correspond to string language operations. \par
First, in Section~\ref{sec:stringcaseops}, we investigate operations that correspond to string language operations. These are the Boolean operations of union, intersection and complementation, and also two special cases of (inverse) morphisms: letter-to-letter morphisms, and more special surjective  letter-to-letter morphisms, known as  
\emph{projections} in the terminology of array languages, and, more generally, the two-dimensional morphisms as defined in Section~\ref{sec:2DPatternLanguages}\par
Next, in Section~\ref{sec:similartostringcaseops}, we take a closer look at operations similar to string language operations. More precisely, we investigate closure under concatenation and concatenation closure (or Kleene star), which constitute classical operations for string languages, but, with respect to the array case, we encounter an important difference, namely, there are two different types of concatenations: row and column concatenation. In particular, the concatenation of two arrays could be undefined (just because the dimensions do not match), but the concatenation of the two corresponding languages need not be empty.\par
Finally, in Section~\ref{sec:arraycaseops}, we investigate \emph{operations special to arrays}, that are usually not considered or even defined for string languages. These are mainly geometric operations like quarter turn, half turn, reflection or transposition of an array.

\subsection{String Language Operations}\label{sec:stringcaseops}

We first point out that, due to Lemma~\ref{lem-transfer} below, whenever a non-closure result is known for terminal-free non-erasing string pattern languages, this would straightforwardly transfer to the array case. We will therefore focus on finding proofs for the string case for non-closure properties, and conversely, we will try to give proofs for the array case for closure properties.
Interestingly enough, (non-)closure properties have not been studied for the (classical)  terminal-free non-erasing string pattern languages, all published
proofs that we are aware of for this topic use terminals or erasing. So, our study also contributes to the theory of  string pattern languages.
Conversely, if we do not manage to find proofs or examples as required for the mentioned approach, this implicitly always raises an open classical string language question. For any mode $z \in \{h, p, r, c, rc\}$ and any pattern $\pi$, let $L_{\Sigma, z}^{\oneD}(\pi)$ denote those arrays from $L_{\Sigma, z}(\pi)$ that have just one row, i.\,e., $L_{\Sigma, z}^{\oneD}(\pi) := \{W \in L_{\Sigma, z}(\pi) \mid |W|_r = 1\}$. Clearly, such arrays can be interpreted as strings and vice versa. In this sense, $L_{\Sigma, z}^{\oneD}(\pi)$ and the string language $L^{\oneD}_{\Sigma}(\pi)$ generated by the pattern $\pi$ coincide, as long as $z\neq h$. For $z=h$, we encounter the special case that all inserted words have to be of the same length. Let us formulate this more formally:

\begin{lemma}\label{lem-transfer}
Let $\pi$ be an array pattern of height one.
Then, $\pi$ is, at the same time, a string pattern.
Moreover, for any $z \in \{p, r, c, rc\}$, 
$L_{\Sigma, z}^{\oneD}(\pi)=L^{\oneD}_{\Sigma}(\pi)$, while
 $L_{\Sigma, h}^{\oneD}(\pi)\subseteq L^{\oneD}_{\Sigma}(\pi)$.
\end{lemma}

We shall now prove non-closure properties for $L^{\oneD}_{\Sigma}(\pi)$, which directly carry over to the classes $L_{\Sigma, z}(\pi)$, $z \in \{h, p, r, c, rc\}$ (for some operations, however, the class $L_{\Sigma, h}(\pi)$  constitutes a special case, which is treated separately). To this end, we will mostly focus on two patterns: $\alpha=xyx$ and $\beta=xxy$. The next lemma states an immediate observation for these patterns.

\begin{lemma}
Over the terminal alphabet $\Sigma=\{\ta,\tb\}$, 
let $L_s(\alpha)$ ($L_s(\beta)$) denote
the shortest words that can be described by $\alpha$ and $\beta$, respectively, disallowing erasing.
Then,
\begin{eqnarray*}
L_s(\alpha)&=&\{\ta \ta \ta, \ta \tb \ta, \tb \ta \tb, \tb \tb \tb\}\,,\\
L_s(\beta)&=&\{\ta \ta \ta, \ta \ta \tb, \tb \tb \ta, \tb \tb \tb\}\,.
\end{eqnarray*}
\end{lemma}

\begin{proposition}\label{prop-union}
 For any non-unary alphabet $\Sigma$, $\mathcal{L}^{\oneD}_{\Sigma}$ is not closed under union.
\end{proposition}

\begin{proof}
Without loss of generality, $\{\ta, \tb\}\subseteq\Sigma$. In the following argument, we actually focus on 
$\Sigma=\{\ta, \tb\}$, but this can be easily extended to the more general case. 
Assume that there was a pattern $\gamma$ with $\mathcal{L}^{\oneD}_{\Sigma}(\gamma)=\mathcal{L}^{\oneD}_{\Sigma}(\alpha)\cup \mathcal{L}^{\oneD}_{\Sigma}(\beta)$.
Then, $$L_s(\gamma)=L_s(\alpha)\cup L_s(\beta)=\{\ta \ta \ta, \ta \ta \tb, \ta \tb \ta, \tb \ta \tb, \tb \tb \ta, \tb \tb \tb\}.$$
This means that $\gamma$ contains exactly three variable occurrences (with more, these words cannot be generated, with less, shorter
words could be generated), i.e., $\gamma=y_1\,y_2\,y_3$, $y_i \in X$, $1 \leq i \leq 3$. As $L_s(\gamma)\neq\{\ta, \tb\}^3$, some of these variables must coincide, which leads to a contradiction. More precisely, if $y_1 = y_2$ or $y_2 = y_3$, then $\ta \tb \ta$ cannot be generated and if $y_1 = y_3$, then $\ta \ta \tb$ cannot be generated.
Hence, $\gamma$ with $L_s(\gamma)=L_s(\alpha)\cup L_s(\beta)$  cannot exist.
\end{proof}

Now if there was a pattern $\gamma$ such that  $L_{\Sigma, z}(\gamma)= L_{\Sigma, z}(\alpha)\cup L_{\Sigma, z}(\beta)$, $z \in \{p, r, c, rc\}$, then, by Lemma~\ref{lem-transfer}, this would imply $L_{\Sigma}^{\oneD}(\gamma)= L_{\Sigma}^{\oneD}(\alpha)\cup L_{\Sigma}^{\oneD}(\beta)$, contradicting Proposition~\ref{prop-union}. We point out that in the proof of Proposition~\ref{prop-union}, we do not use any replacements by words of different lengths to obtain our contradiction. Hence, this argument is also valid in the case when $z=h$. 

\begin{corollary}\label{cor-clos-union}
None of the array pattern language classes under consideration (over some non-unary alphabet) is closed under union.
\end{corollary}

We proceed with the intersection operation.

\begin{proposition}\label{prop-intersection} For any non-unary alphabet $\Sigma$, 
 $\mathcal{L}^{\oneD}_{\Sigma}$ is not closed under intersection.
\end{proposition}

\begin{proof}
The argument resembles the previous proof.
Assume that $\gamma$ describes $L^{\oneD}_{\Sigma}(\alpha)\cap L^{\oneD}_{\Sigma}(\beta)$.
Notice that $L_s(\gamma)=\{\ta \ta \ta, \tb \tb \tb\}$, which clearly implies that $\gamma=xxx$.
However, $\ta \ta \tb \ta \ta\in (L^{\oneD}_{\Sigma}(\alpha)\cap L^{\oneD}_{\Sigma}(\beta))\setminus L^{\oneD}_{\Sigma}(\gamma)$.
\end{proof}

Notice that the replacement words we used for deriving a contradiction are of different lengths, meaning that $\ta \ta \tb \ta \ta \in L_{\Sigma}(\alpha)$ because of the replacement $x\mapsto \ta \ta$ and $y\mapsto \tb$, but
$\ta \ta \tb \ta \ta\in L_{\Sigma}(\beta)$ because of $x\mapsto \ta$ and $y\mapsto \tb \ta \ta$. Hence, we cannot conclude non-closure for the $h$-mode in the following corollary:

\begin{corollary}
None of the array pattern language classes under consideration (over some non-unary alphabet and apart from the $h$-case) is closed under intersection.
\end{corollary}

Indeed, the $h$-mode plays a special r\^ole, as can be seen by the following 
result.

\begin{proposition}\label{prop-hom-intersection}Let $\Sigma$ be some alphabet.
Then, $\mathcal{L}_{\Sigma, h}$ 
is closed under intersection.
 \end{proposition}

\begin{proof}
Assume that $a\in\Sigma$. 
Let $\alpha,\beta$ be two array patterns. Let $m_\alpha$ be the height (number of rows) of $\alpha$ and $n_\alpha$ be the width (number of columns) of $\alpha$. 
Likewise, $m_\beta$ and $n_\beta$ are understood.
Then, the width of the smallest arrays in  $L_{\Sigma, h}(\alpha)\cap
 L_{\Sigma, h}(\beta)$ equals $n=\operatorname{lcm}(n_\alpha,n_\beta)$,
and their height equals $m=\operatorname{lcm}(m_\alpha,m_\beta)$.
This can be easily seen by substituting each variable in $\alpha$
by the unique array of width $n'_\alpha=\frac{\operatorname{lcm}(n_\alpha,n_\beta)}{n_\alpha}$ and height  $m'_\alpha=\frac{\operatorname{lcm}(m_\alpha,m_\beta)}{m_\alpha}$
 over the alphabet $\{a\}$
into the pattern $\alpha$, as $n=n_\alpha \times n'_\alpha$,
$m=m_\alpha \times m'_\alpha$, and a similar substitution within $\beta$.

For each variable $x$ that occurs in $\alpha$, take $m'_\alpha \times n'_\alpha$ new variables
$x_{i,j}$ with $1\leq i\leq m'_\alpha$ and $1\leq j\leq n'_\alpha$.
Define a morphism $h_\alpha$ by replacing the variable $x$ by the following array  of height $n'_\alpha$ and width $m'_\alpha$, consisting of the previously introduced variables:
$$\begin{smallmatrix}
   x_{11} & x_{12} & \cdots & x_{1n'_\alpha}\\
   x_{21} & x_{22} & \cdots & x_{2n'_\alpha}\\
\vdots & \vdots & \vdots & \vdots\\
x_{m'_\alpha 1} & x_{m'_\alpha 2} & \cdots & x_{m'_\alpha n'_\alpha}
  \end{smallmatrix}
$$
Hence, $h_\alpha(\alpha)$ is some array of height $m$ and of width $n$.
Accordingly, one can define a morphism $h_\beta$ such that $h_\beta(\beta)$ is again some array of height $m$ and of width $n$.
Due to Lemma~\ref{lem-hom-closure}, $L_{\Sigma, h}(h_\alpha(\alpha))\subseteq L_{\Sigma, h}(\alpha)$ and $L_{\Sigma, h}(h_\beta(\beta))\subseteq L_{\Sigma, h}(\beta)$ (*).
Namely, if $U\in L_{\Sigma, h}(h_\alpha(\alpha))$, then there exists some two-dimensional morphism $g$ such that $U=g(h_\alpha(\alpha))$, i.e., there also
exists some two-dimensional morphism $f:=h_\alpha\circ g$ with $U=f(\alpha)$, so that $U\in L_{\Sigma, h}(\alpha)$. 
Now, define an array pattern $\gamma$ of height $m$ and of width $n$, consisting exclusively of variable entries, as follows:
The variables occurring at positions $(i,j)$ and $(i',j')$ in $\gamma$ (where $1\leq i,i'\leq m$ and $1\leq j,j'\leq n$) are identical if and only if the corresponding variables in at least one of the arrays $h_\alpha(\alpha)$ or  $h_\beta(\beta)$ are identical.

We claim that $L_{\Sigma, h}(\gamma)=L_{\Sigma, h}(\alpha)\cap
 L_{\Sigma, h}(\beta)$.

As $\gamma$ is obtained from $h_\alpha(\alpha)$ by identifying certain variables, due to (*) we find that:
$$L_{\Sigma, h}(\gamma)\subseteq L_{\Sigma, h}(h_\alpha(\alpha))\subseteq L_{\Sigma, h}(\alpha),$$
and likewise for $\beta$, so that $L_{\Sigma, h}(\gamma)\subseteq L_{\Sigma, h}(\alpha)\cap
 L_{\Sigma, h}(\beta)$.

Conversely, we have already argued that the smallest arrays in $L_{\Sigma, h}(\alpha)\cap
 L_{\Sigma, h}(\beta)$ have height $m$ and width $n$.
More generally, any array $U\in L_{\Sigma, h}(\alpha)\cap
 L_{\Sigma, h}(\beta)$ has height $k\cdot m$ and width $l\cdot n$.
As $U\in L_{\Sigma, h}(\alpha)$, there is some two-dimensional morphism $h_{U,\alpha}$ such that 
$U=h_{U,\alpha}(\alpha)$. Moreover, for each variable $x$ in $\alpha$, $|h_{U,\alpha}(x)|_r=k_\alpha=\frac{k\cdot m}{m_\alpha}$
and  $|h_{U,\alpha}(x)|_c=l_\alpha=\frac{l\cdot n}{n_\alpha}$. We can make an analogous reasoning with $\beta$, introducing
the constants $k_\beta=\frac{k\cdot m}{m_\beta}$ and $l_\beta=\frac{l\cdot n}{n_\beta}$ for the morphism $h_{U,\beta}$. 
As $U=h_{U,\alpha}(\alpha)=h_{U,\beta}(\beta)$, entries in $U$ must coincide both according to $\alpha$ and to $\beta$.
This is exactly reflected in the construction of $\gamma$ provided above, so that $h_{U,\alpha}=h_{\alpha}\circ h'$
and $h_{U,\beta}=h_{\beta}\circ h'$ for some two-dimensional morphism $h'$ with $h'(\gamma)=U$.
Hence, $U\in L_{\Sigma, h}(\gamma)$.
\end{proof}

Arguments as in Propositions~\ref{prop-union} and~\ref{prop-intersection}
can be given for any non-trivial binary set operation, for instance, symmetric difference or set difference.
This also gives the according result for complementation, but there is also an easier argument in that case.
Notice that, as non-erasing pattern languages or array patterns cannot reasonably cope with the empty word or the empty array,
we disregard this in the complement operation.

\begin{proposition}\label{prop-complementation}
For any  alphabet $\Sigma$, $\mathcal{L}^{\oneD}_{\Sigma}$ is not closed under complementation.
\end{proposition}

\begin{proof}
Let $\Sigma$ be some alphabet with $a\in\Sigma$. 
Consider the pattern $\alpha=xy$.
The complement (disregarding the empty word) contains as shortest words a word of the form $a$.
This implies that a pattern $\beta$  describing the complement must have length one, i.e., $\beta=x$.
Hence, $L(\beta)=\Sigma^+$, which is not the complement of $L(\alpha)$.
\end{proof}

\begin{corollary}
None of the array pattern language classes under consideration (over any alphabet) is closed under complementation.
\end{corollary}

Notice that in the other cases (but complementation),  we cannot cope with unary alphabets.
This might need some different arguments.\par
We shall now turn to operations that are described by different kinds of morphisms. For the array case, codings (or projections) is a common such operation.

\begin{theorem}\label{thm-projection}
Any of our array pattern language classes (over arbitrary alphabets) is closed under projections.
\end{theorem}

\begin{proof}
The following proof works for every $z \in \{h, p, r, c, rc\}$. Let $\alpha$ be some array pattern and let $\pi:\Sigma\to\Gamma$ be a surjective mapping that describes some projection. We shall now show that $\pi(L_{\Sigma, z}(\alpha))=L_{\Gamma, z}(\alpha)$.\par
Namely, consider some array $U\in L_{\Sigma, z}(\alpha)$.
$U$ is obtained from $\alpha$ by replacing any variable $x$ occurring in $\alpha$ by some
array $h(x)$. If we replace $x$ by $\pi(h(x))$ instead (which has the same dimensions as $h(x)$),
we can see that $\pi(U)\in L_{\Gamma, z}(\alpha)$.\par
Conversely, fix some letter from $\pi^{-1}(\ta)$ for each $\ta\in\Gamma$ for now and denote it by $\kappa(\ta)$.
This is possible, as $\pi$ is a surjective mapping.
As $\pi(\kappa(\ta))=\ta$ for each $\ta\in\Gamma$ by construction, for any $U\in L_{\Gamma, z}(\alpha)$ we can 
describe some $V\in L_{\Sigma, x}(\alpha)$ such that $\pi(V)=U$, just by taking $V=\kappa(U)$.
\end{proof}

The result does not generalize to (string) morphisms where each image is of the same length.

\begin{proposition}\label{prop-length2-morphisms} For any non-unary alphabet $\Sigma$, 
 $\mathcal{L}^{\oneD}_{\Sigma}$ is not closed under
morphisms that map every letter to a word of length two.
\end{proposition}

\begin{proof}
Consider the pattern $x$ on the alphabet $\Sigma=\{\ta, \tb\}$, describing the universal language $\Sigma^+$, and the morphism $h$ with $h(\ta)=\ta \tb$ and $h(\tb)=\ta \tb$. Then, $h(\Sigma^+)=\{\ta \tb\}^+$, which is not a pattern language, as an easy analysis shows.
\end{proof}

\begin{corollary}
None of the array pattern language classes under consideration (over some non-unary alphabet) is closed under two-dimensional morphisms.
\end{corollary}

This is also true for the more general operation of substitution, with the same examples.

Let us also remark that Proposition~\ref{prop-length2-morphisms}
did not rely on the fact that we restricted our attention to one specific non-unary alphabet $\Sigma$.
However, if we have a specific alphabet, then we can even state:

\begin{proposition}\label{prop-bijective-coding}
 Any of our array pattern language classes (over some fixed  alphabet $\Sigma$) is closed under some projection $\pi:\Sigma\to\Sigma$ 
if and only if $\pi$ is a bijection.
\end{proposition}

\begin{proof}
If $\pi$ is some bijection, then the argument given in the proof of Theorem~\ref{thm-projection} applies. Recall now that for finite sets as $\Sigma$, $\pi$ is a bijection if and only if $\pi$ is a surjection. If $\pi$ is not a bijection, then $\pi$ is not a surjection; hence, there is some $b \in \Sigma\setminus\pi(\Sigma)$, Consider the pattern $\alpha=x$. Clearly, $\pi(L_{\Sigma, z}(x))=(\pi(\Sigma))^{**}$, $z \in \{h, p, r, c, rc\}$, but any array pattern language over $\Sigma$ will also contain patterns with the letter $b$.
\end{proof}

This immediately implies the following two results:

\begin{corollary}
Any of our  array pattern language classes (over
binary alphabets) is closed under conjugation.
\end{corollary}

\begin{corollary}
None of our  array pattern language classes (over some fixed  alphabet $\Sigma$) is closed under all letter-to-letter morphisms.
\end{corollary}

Alternatively, we could also look at inverse morphisms. Here, we already get negative results for inverse codings.

\begin{proposition}\label{prop-one-dim-inverse-coding} For any alphabet $\Sigma$ with at least four letters, 
 $\mathcal{L}^{\oneD}_{\Sigma}$ is not closed under
inverse letter-to-letter morphisms.
\end{proposition}

\begin{proof}
Consider the coding
$h: \ta \mapsto 1, \tb\mapsto 1, \tc\mapsto 2, \td\mapsto 3$. The pattern $xx$ generates (over the alphabet $\{1,2,3\}$) the language $L=\{ww\mid w\in \{1,2,3\}^+\}$. However, $L'=h^{-1}(L)$ is not a pattern language. The shortest words in $L'$ are of length two, so that a hypothetical pattern for $L'$ must be of the form $yz$ or $yy$. In the first case, $\tc \td\notin L'$ is produced, while in the second case, $\ta \tb\in L'$ cannot be produced.
\end{proof}

\begin{corollary}
None of the array pattern language classes under consideration (over some sufficiently large alphabet) is closed under inverse (two-dimensional) morphisms.
\end{corollary}

\subsection{Operations Similar to String Language Operations}\label{sec:similartostringcaseops}

Let us first turn to the concatenation operation.
As a warm-up, we first consider the string case.

\begin{lemma}\label{lem-string-concatenation} For any alphabet $\Sigma$, 
 $\mathcal{L}^{\oneD}_{\Sigma}$ is closed under concatenation.
\end{lemma}

\begin{proof}
Consider two patterns $\alpha$ and $\beta$. After renaming, we can assume that $\alpha$ and $\beta$ do not contain any
identical variables.
Then, $L^{\oneD}_{\Sigma}(\alpha\beta)=L^{\oneD}_{\Sigma}(\alpha)L^{\oneD}_{\Sigma}(\beta)$.
\end{proof}

The first thing one should note is that the concatenation of two arrays could be undefined (i.\,e., if their dimensions do not match), even though the concatenation of the two according languages need not be empty. However, we can prove:

\begin{theorem}\label{thm-catenation} Fix some alphabet $\Sigma$.
\begin{itemize}
 \item $\mathcal{L}_{\Sigma,r}$ is closed under row concatenation $\varominus$;
\item  $\mathcal{L}_{\Sigma,c}$ is closed under column concatenation $\varobar$;
\item  $\mathcal{L}_{\Sigma,p}$  and  $\mathcal{L}_{\Sigma,h}$ are closed both under row and under column concatenation.
\end{itemize}
\end{theorem}

\begin{proof}
We only prove the first item. The others follow similarly. Observe that in the case of $\mathcal{L}_{\Sigma,h}$, we need Lemmas~\ref{lem-hom-charac} and~\ref{lem-hom-closure} to finish the argument. Let $\alpha$ and $\beta$ be two array patterns. We can assume that the variable alphabets of $\alpha$ and $\beta$ are disjoint.
We want to construct an array pattern $\gamma$ such that
$$L_ {\Sigma,r} (\gamma)=L_ {\Sigma,r}(\alpha)\varominus L_ {\Sigma,r}(\beta).\qquad (*)$$
Let $m_\alpha$ be the number of rows of $\alpha$ and $n_\alpha$ be the number of columns of $\alpha$.
Accordingly, $m_\beta$ is the number of rows of $\beta$ and $n_\beta$ is the number of columns of $\beta$.
If $n_\alpha=n_\beta$, then we can set $\gamma=\alpha\varominus\beta$ to satisfy $(*)$ as in the string case of Lemma~\ref{lem-string-concatenation}.
More generally, set $n=\operatorname{lcm}(n_\alpha,n_\beta)$.
$n$ is the width of the smallest arrays in $L_ {\Sigma,r}(\alpha)\varominus L_ {\Sigma,r}(\beta)$.
More generally speaking, any array in $L_ {\Sigma,r}(\alpha)\varominus L_ {\Sigma,r}(\beta)$ has some width
that is a multiple of $n$. We are going to exploit this property by constructing two array patterns $\alpha'$
and $\beta'$ of width $n$ such that
$$L_ {\Sigma,r}(\alpha)\varominus L_ {\Sigma,r}(\beta)
=L_ {\Sigma,r}(\alpha')\varominus L_ {\Sigma,r}(\beta'),$$
so that we can apply our previous reasoning, defining now $\gamma=\alpha'\varominus\beta'$.
Consider the two-dimensional morphism $h_\alpha$ that maps every variable $x$ of $\alpha$ on an array of height one and 
 width $n'_\alpha=\frac{\operatorname{lcm}(n_\alpha,n_\beta)}{n_\alpha}$,
more precisely, onto $x_1x_2\cdots x_{n'_\alpha}$.
Here, $x_1,\dots,x_{n'_\alpha}$ are ``new variables.''
Then, $h_\alpha(\alpha)$ is an array pattern of width $n=n_\alpha\cdot n'_\alpha$.
Similarly, we define a morphism $h_\beta$, yielding an array pattern $h_\beta(\beta)$ of width $n=n_\beta\cdot n'_\beta$.
Now, any array $U$ in $L_ {\Sigma,r}(\alpha')$ has a width that is a multiple of $n$ and also belongs to $L_ {\Sigma,r}(\alpha)$,
and conversely any $U\in L_ {\Sigma,r}(\alpha)$ that  has a width that is a multiple of $n$ belongs to  $L_ {\Sigma,r}(\alpha')$.
Together with similar statements for the pattern $\beta$, the claim follows.
\end{proof}

It is not a coincidence that for  $\mathcal{L}_{\Sigma,r}$ and  $\mathcal{L}_{\Sigma,c}$,
we had to focus on the ``correct'' concatenation operation in the preceding theorem.
More precisely, we can show:

\begin{theorem}
\label{thm-concatenation-nonclosure}
 Fix some non-unary alphabet $\Sigma$.
\begin{itemize}
 \item $\mathcal{L}_{\Sigma,r}$ is not closed under column  concatenation $\varobar$;
\item  $\mathcal{L}_{\Sigma,c}$ is not closed under row concatenation $\varominus$;
\item  $\mathcal{L}_{\Sigma,rc}$  is neither closed under row nor under column concatenation.
\end{itemize}
\end{theorem}

\begin{proof}
Again, we only prove the first item; the others can be seen in a similar fashion.

Consider the array patterns $\alpha=\begin{bsmallmatrix}
x_1 & x_2\\
x_2 & x_3
\end{bsmallmatrix}$
and $\beta=\begin{bsmallmatrix}
                             y_1 & y_2\\
y_2 & y_3
                            \end{bsmallmatrix}$.
Notice that $L_{\Sigma,r}(\alpha)=L_{\Sigma,r}(\beta)$,
There are arrays of width four and height two in $L:=L_{\Sigma,r}(\alpha)\varobar L_{\Sigma,r}(\beta)$,
and these are the smallest arrays in $L$. Hence, any array pattern $\gamma$ with $L=L_{\Sigma,r}(\gamma)$ has width four and height two.
Let $\gamma=\begin{bsmallmatrix}
             z_1 & z_2 & z_3 & z_4\\
 z_5 & z_6 & z_7 & z_8
            \end{bsmallmatrix}
$, with possibly some of the variables being the same.
As $\begin{bsmallmatrix}
             \ta & \tb & \ta & \tb\\
             \tb & \ta & \tb & \ta \\
            \end{bsmallmatrix}
\in L
$, but $\begin{bsmallmatrix}
             \ta & \tb & \ta & \ta\\
             \tb & \ta & \tb & \ta \\
            \end{bsmallmatrix}
\notin L
$ and also 
$\begin{bsmallmatrix}
             \ta & \ta & \ta & \tb\\
             \tb & \ta & \tb & \ta \\
            \end{bsmallmatrix}
\notin L
$, $z_2=z_5$ and $z_4=z_7$.
More generally, it can be verified that the array pattern
$\zeta=\begin{bsmallmatrix}
             z_1 & z_2 & z_3 & z_4\\
 z_2 & z_6 & z_4 & z_8
            \end{bsmallmatrix}
$ describes those and only those arrays of width four and height two that belong to $L$.
However, $U=\begin{bsmallmatrix}
           \ta & \ta & \tb & \ta & \tb & \tb\\
\tb & \ta & \tb & \tb& \ta & \tb
          \end{bsmallmatrix}$
is an array from $L_{\Sigma,r}(\zeta)$ (namely, consider $z_1=\ta\ta$, $z_2=\tb\ta$, $z_3=z_4=z_6=\tb$  and $z_8=\ta\tb$) that does not belong to $L$, as this would mean that
$U=U_1\varobar U_2$ for some array $U_1$ or $U_2$ of width two and height two from  $L_{\Sigma,r}(\alpha)=L_{\Sigma,r}(\beta)$.
However, neither $U_1=\begin{bsmallmatrix}
                       \ta & \ta\\
\tb & \ta
                      \end{bsmallmatrix}
$
 nor $U_2=\begin{bsmallmatrix}
           \tb & \tb\\
\ta & \tb
          \end{bsmallmatrix}
$
belongs to $L_{\Sigma,r}(\alpha)=L_{\Sigma,r}(\beta)$.
\end{proof}

Notice that the proofs of negative closure properties necessitate a non-unary alphabet to work.

We now turn to the Kleene closure. Here, we can again first show a non-closure result for the string case that then readily transfers to the array cases.

\begin{lemma}
 \label{lem-Kleene}
Let $\Sigma$ be a non-unary alphabet.
Consider $\alpha=xx$. Then, $(L^{\oneD}_{\Sigma}(\alpha))^+\notin \mathcal{L}^{\oneD}_{\Sigma}$.
\end{lemma}

\begin{proof}
The shortest words in $L:=(L^{\oneD}_{\Sigma}(\alpha))^+$ are of length two. Hence, there are only two different possibilities for any pattern $\beta$ with $L^{\oneD}_{\Sigma}(\beta)=L$: If $\beta=\alpha=xx$, then $\ta\ta\tb\tb\in L\setminus L^{\oneD}_{\Sigma}(\beta)$, while if $\beta=xy$, then $\ta\tb\in L^{\oneD}_{\Sigma}(\beta)\setminus L$.
\end{proof}

\begin{proposition}\label{prop-Kleene}
Let $\Sigma$ be a non-unary alphabet.
Then, none of the array language families 
  $\mathcal{L}_{\Sigma,x}$ with $x\in\{r,c,rc,p,h\}$ is closed under 
column concatenation closure nor under row concatenation closure.
\end{proposition}

\begin{proof}
 Consider  $\mathcal{L}_{\Sigma,r}$. Due to Lemmas~\ref{lem-transfer} and~\ref{lem-Kleene}, this class is not closed under row concatenation closure.
For the case of column contentation closure, reconsider the proof of Theorem~\ref{thm-concatenation-nonclosure}.
There, we have presented a language $L\in\mathcal{L}_{\Sigma,r}$ such that $L\varobar L\notin \mathcal{L}_{\Sigma,r}$.
But that argument also shows that the column contentation closure of $L$ does not belong to $\mathcal{L}_{\Sigma,r}$.

The other cases are simillarly seen.
For the case of morphisms, observe that the contradiction in Lemma~\ref{lem-Kleene} was derived by
substituting the variables by words of the same length.
\end{proof}

\subsection{Operations Special to Arrays}\label{sec:arraycaseops}

Recall that the transposition operation is first defined for arrays (or patterns) and can then be lifted to languages and even to language classes.
Nearly by definition, we find:

\begin{lemma}\label{lem-transpose}
Let $\Sigma$ be some alphabet. Let $\alpha$ be a pattern. Then, 
 $\transpose{L_{\Sigma, r}(\alpha)}=L_{\Sigma, c}(\transpose{\alpha})$
and   
 $\transpose{L_{\Sigma, c}(\alpha)}=L_{\Sigma, r}(\transpose{\alpha})$.
\end{lemma}

\begin{corollary}\label{cor-transpose}
Let $\Sigma$ be some alphabet. Then, 
  $\transpose{\mathcal{L}_{\Sigma, r}}= \mathcal{L}_{\Sigma, c}$ and
  $\transpose{\mathcal{L}_{\Sigma, c}}= \mathcal{L}_{\Sigma, r}$.
\end{corollary}

Since $\alpha :=
\begin{bsmallmatrix}
  x_1 & x_2 \\
  x_2 & x_1 \\
 \end{bsmallmatrix}$ is identical to its transposition and, as shown in the proof of Lemma~\ref{prcIncompLemma}, describes an $r$ pattern language (a $c$ pattern language), which is not a $c$ pattern language (not an $r$ pattern language, respectively), we can conclude the following:

\begin{proposition}\label{prop-transposition}
Let $\Sigma$ be an alphabet.
Neither $\mathcal{L}_{\Sigma, c}$ nor $\mathcal{L}_{\Sigma, r}$
are closed under transposition.
\end{proposition}

\begin{proof}
Consider the pattern $\alpha :=
\begin{bsmallmatrix}
  x_1 & x_2 \\
  x_2 & x_1 \\
 \end{bsmallmatrix}$. 
As $\alpha=\transpose{\alpha}$,  $$\transpose{L_{\Sigma, r}(\alpha)}=L_{\Sigma, c}(\transpose{\alpha})=L_{\Sigma, c}(\alpha)\notin \mathcal{L}_{\Sigma, r},$$
 as we have shown in Lemma~\ref{prcIncompLemma}.
Symmetrically, the other claim follows.
\end{proof}

\begin{proposition}\label{prop-transposition-a}
For any alphabet $\Sigma$ and $x\in\{h,p,rc\}$, $\mathcal{L}_{\Sigma, x}$ is closed under transposition.
 \end{proposition}

\begin{proof}
For $h$ and $p$, this claim is immediate from the fact that we have proper factorizations.
For the case $x=rc$, we use Lemma~\ref{lem-transpose}.
 Let $\alpha$ be some pattern.
Then,
\begin{eqnarray*}
\transpose{L_{\Sigma, rc}(\alpha)}&=&\transpose{({L_{\Sigma, r}(\alpha)}\cup {L_{\Sigma, c}(\alpha)})}\\
&=&L_{\Sigma, c}(\transpose{\alpha})\cup L_{\Sigma, r}(\transpose{\alpha})\\
&=& L_{\Sigma, rc}(\transpose{\alpha})
\end{eqnarray*}
This immediately implies the claim.
\end{proof}

With respect to purely geometric operations as turns and reflections, we find the following:

\begin{proposition}\label{prop-turns}
Let $\Sigma$ be some alphabet.
\begin{itemize}
\item $\mathcal{L}_{\Sigma,rc}$, $\mathcal{L}_{\Sigma,p}$ and $\mathcal{L}_{\Sigma,h}$ are closed under quarter-turn.
\item For every $x\in\{r,c,rc,p,h\}$, $\mathcal{L}_{\Sigma,x}$ is closed under half-turn and reflections.
\item $\mathcal{L}_{\Sigma,r}$ and $\mathcal{L}_{\Sigma,c}$ are closed 
neither under left nor under right turn.
\end{itemize}
\end{proposition}

\begin{proof}
For the positive closure results, simply observe that the language described by the quarter-turn, by the half-turn or by a reflection of the array pattern $\alpha$ is just the quarter-turn, the half-turn or the reflection of the language described by $\alpha$.\par
For the non-closure properties, by symmetry it suffices to show that there is a language in  $\mathcal{L}_{\Sigma,r}$
whose quarter-turn is not in $\mathcal{L}_{\Sigma,r}$.
To this end, consider $L:=L_{\Sigma,r}\left(\begin{bsmallmatrix}x&y\\y&x\end{bsmallmatrix}\right)$.
Observe that the quarter-turn of $L$ is the same as $L_{\Sigma,c}\left(\begin{bsmallmatrix}x&y\\y&x\end{bsmallmatrix}\right)$,
which was proven not to be in $\mathcal{L}_{\Sigma,r}$ in   Lemma~\ref{prcIncompLemma}.
\end{proof}

The positive closure properties can be easily observed by applying the geometric operation directly on the array pattern. In order to show non-closure of $\mathcal{L}_{\Sigma,r}$ and $\mathcal{L}_{\Sigma,c}$ with respect to left and right turn, it is again sufficient to observe that the pattern $\alpha$ from above is identical to its left or right turn and then apply a similar argument as in the proof of Lemma~\ref{prcIncompLemma}.\par
Due to symmetry, it does not matter if we consider horizontal or vertical reflections. Notice that both half-turns and reflections coincide in the string case in any meaningful, non-trivial interpretation; in that case,
the operation is also known as mirror image.

\section{Future Research Directions}

A thorough investigation of the typical decision problems for two-dimensional pattern languages like the membership, inclusion and equivalence problem is left for future research. It can be easily seen that the NP-completeness of the membership problem for string pattern languages carries over to $\mathcal{L}_{\Sigma,x}$, $x \in \{p, r, c, rc\}$. On the other hand, for a given array pattern $\alpha$ and a terminal array $W$, the question whether or not $W \in L_{\Sigma, h}(\alpha)$ can be decided in polynomial time by checking whether $W$ is a morphic image of $\alpha$ with respect to a $(\frac{|W|_r}{|\alpha|_r}, \frac{|W|_c}{|\alpha|_c})$-uniform substitution. As shown by Lemma~\ref{equalityLemma}, the equivalence problem for all the classes $\mathcal{L}_{\Sigma,x}$ with $x \in \{h, p, r, c, rc\}$ and $|\Sigma| \geq 2$ can be easily solved by simply comparing the patterns. However, for every $z, z' \in \{h, p, r, c, rc\}$, $z \neq z'$, the problem to decide for given patterns $\alpha$ and $\beta$ whether or not $L_{\Sigma, z}(\alpha) = L_{\Sigma, z'}(\beta)$ might be worth investigating. The inclusion problem for terminal-free nonerasing string pattern languages is still open. Hence, with respect to the inclusion problem, a positive decidability result for two-dimensional pattern languages implies a positive decidability result for terminal-free nonerasing string pattern languages.\par
For string pattern languages it is common to use terminal symbols in the patterns as well as to consider the \emph{erasing} case, i.\,e., variables can be replaced by the empty word. The $p$ pattern languages can be adapted to the erasing case by allowing variables to be substituted by the empty array. Furthermore, the situation of having a terminal symbol at position $(i, j)$ of an array pattern simply forces all the variables in the $i^{\text{th}}$ row to be substituted by arrays of height $1$ and all the variables in the $j^{\text{th}}$ column to be substituted by arrays of width $1$. As in the string case, it is likely that in the two-dimensional case the difference between erasing and nonerasing substitutions and patterns with and without terminal symbols lead to different language classes with different decidability properties, too. \par
Finally, we wish to point out that it is straightforward to generalise our different classes of two-dimensional pattern languages to the three-dimensional or even $n$-dimensional case.

\bibliographystyle{plain}
\bibliography{abbrev,MSbib}

\end{document}